\documentclass[12pt,a4paper]{article}
\usepackage{amsmath}
\usepackage{amsfonts}
\usepackage{amssymb}
\usepackage{amsthm}

\input{epsf}             % encapsulated Postscript figures
\newcommand{\Spin}{{\mathrm {Spin}}}

\newcommand{\SU}{{\mathrm {SU}}}

\newcommand{\SO}{{\mathrm {SO}}}
\newcommand{\so}{{\mathrm {so}}}

\newcommand{\dd}{{\mathrm d}}      % \d already defined
\DeclareMathOperator{\tr}{Tr}  %% this is the official way - should change the rest

\newcommand{\R}{{\mathbb R}}   % Real numbers
\newcommand{\C}{{\mathbb C}}   % Complex numbers
\newcommand{\HH}{{\mathbb H}}   % Quaternions
\newcommand{\Z}{{\mathbb Z}}   % Integers
\newcommand\End{\mbox{End}}
\renewcommand\Re{\operatorname{Re}} % real part
\DeclareMathOperator{\Lie}{Lie}
\newcommand{\SpinOperator}{{S}} %symbol for spin Lie algebra generators

\theoremstyle{definition}
 
\newtheorem{lemma}{Lemma} 
\newtheorem{definition}{Definition} 
\newtheorem{example}{Example} 
\newtheorem{examples}[example]{Examples} 

\begin{document}

\title {Matrix geometries and fuzzy spaces as finite spectral triples}

\author{John W. Barrett
%\thanks{Copyright \copyright\ John W. Barrett 2002}
\\ \\
School of Mathematical Sciences\\
University of Nottingham\\
University Park\\
Nottingham NG7 2RD, UK\\
\\
E-mail john.barrett@nottingham.ac.uk}

\date{}
%\date{Revised July 1st, 2015}

\maketitle

\begin{abstract} A class of real spectral triples that are similar in structure to a Riemannian manifold but have a finite-dimensional Hilbert space is defined and investigated, determining a general form for the Dirac operator. Examples include fuzzy spaces defined as real spectral triples. Fuzzy 2-spheres are investigated in detail, and it is shown that the fuzzy analogues correspond to two spinor fields on the commutative sphere. In some cases it is necessary to add a mass mixing matrix to the commutative Dirac operator to get a precise agreement for the eigenvalues.  
\end{abstract}

\section{Introduction}

A  common theme in geometry is to try to characterise a space in terms of its algebraic structures. 
In Riemannian geometry the functions on a manifold can be multiplied, using the usual pointwise definition
$\bigl(f_1f_2\bigr)(x)=f_1(x)f_2(x).$ This product makes the space of functions an algebra which is commutative, $f_1f_2=f_2f_1$, and infinite-dimensional. Further structures are needed to characterise the geometry completely, for example, specifying the properties of the Dirac operator.

This algebraic way of thinking about geometry is particularly useful in applications to quantum physics, since quantum mechanics is phrased in terms of differential operators. It also suggests the generalisation of geometry to a situation where the analogues of functions or `coordinates' on a space no longer commute; this is called non-commutative geometry. An obvious example is for the phase space of a particle, the functions on phase space being replaced by non-commuting operators in Hilbert space. Another place where this idea is useful is in fermionic systems, where the classical fields anti-commute rather than commute. A more sophisticated application is to the internal space of the standard model, which has a non-commutative structure in which the algebra is a finite-dimensional algebra of matrices \cite{Barrett:2006qq,Connes:2006qv}. 

The fuzzy sphere  \cite{Madore:1991bw} is a simple example of a non-commutative geometry that is an approximation to the commutative 2-sphere $S^2$ with its usual Riemannian metric. The algebra of the fuzzy sphere is the space of $n\times n$ matrices and the geometry was originally specified using an analogue of the Laplace operator. The algebra can be identified (as a vector space) with the finite-dimensional vector space formed by the spherical harmonics on $S^2$ up to some maximum total spin. Thus one can understand the fuzzy sphere as a modification of the commutative sphere geometry so that there is a minimum wavelength for functions on the space. 

This is an interesting feature of the geometry that it is worthwhile trying to generalise. Indeed in quantum gravity there is a fundamental length scale, called the Planck length, and there are no known physical phenomena with a length scale shorter than this length. Thus it seems a fruitful project to build models of geometry in which a minimum wavelength is built in in a fundamental way. The fuzzy sphere can be generalised to fuzzy versions of other spaces on which a symmetry group acts, following the method of \cite{Grosse:1999ci} for $\C\mathrm{P}^2$.

In Riemannian geometry, a spin manifold with a metric is completely determined by the algebra of functions and the Dirac operator. The properties of the Dirac operator can be axiomatised so that an alternative way of thinking about the geometry is that the Dirac operator is the primary variable and the Riemannian metric is derived from it.  The properties of the Dirac operator generalise to the case of non-commutative algebras and so the guiding idea in Connes' approach to non-commutative geometry \cite{Connes:1994yd} is to take the Dirac operator and its associated structures as fundamental.  The axioms for a Dirac operator in non-commutative geometry are called the axioms of a real spectral triple \cite{Connes:1996gi}. An important aspect of the non-commutative case is that the algebra acts twice on the Hilbert space of fermions, once as a left action, and once as a right action that commutes with the left action.

The purpose of this paper is to define and investigate a class of real spectral triples that are similar in structure to Riemannian manifolds but are genuinely non-commutative. These geometries are called here `matrix geometries'. The first requirement is that the algebra of functions on the manifold is replaced by a finite-dimensional matrix algebra $\mathcal A$. On a Riemannian manifold there is also an action of a Clifford algebra in the spinor space $V$ at each point. Therefore, the second requirement of a matrix geometry is that it admits an action of a Clifford algebra. This action should exhaust the degeneracy of the representations of the algebra $\mathcal A$. Thus the matrix geometries are in some sense finite-dimensional approximations to a Riemannian manifold, though a notion of approximation is not investigated here.

The required background material on gamma matrices and their products is given in \S\ref{sec:gamma}. A Clifford algebra is called type $(p,q)$ if it has $p$ generators that square to $+1$ and $q$ generators that square to $-1$; these generators represented as matrices are called gamma matrices. The axioms for spectral triples that have a finite-dimensional Hilbert space are given in \S\ref{sec:spectral}. These are called finite spectral triples and for this case there is a decomposition of the Dirac operator into a part that commutes with the left action of $\mathcal A$ and a part that commutes with the right action. This decomposition is determined by a Frobenius form on $\mathcal A$. 

The notion of a matrix geometry is defined in \S\ref{sec:matrixgeometry}. This general definition is simplified in \S\ref{sec:fuzzygeometries} by considering the fuzzy spaces, which are examples in which the algebra $\mathcal A$ is a simple algebra (the set of all $n\times n$ matrices), and the generalised fuzzy spaces, in which $\mathcal A$ is the direct sum of two simple matrix algebras. For these cases, the structure of the Dirac operator simplifies and so these are most amenable to the computation of features of specific examples. It is apparent in examples that the structure of the Dirac operator is somewhat similar to the case of Riemannian geometry,  with derivatives replaced by commutators and functions replaced by anticommutators.

The final section, \S\ref{sec:spheres}, is devoted to the fuzzy spheres, which are fuzzy versions of the commutative 2-sphere. A Dirac operator for the fuzzy sphere was proposed by Grosse and Presnajder \cite{Grosse:1994ed} using two-dimensional spinors (only the case $k=0$ in their notation is considered here). Let $\sigma^j$, $j=1,2,3$, be $i$ times the $2\times 2$ Pauli matrices. These are type $(0,3)$ gamma matrices. The Hilbert space is $\mathcal H_{GP}=\C^2\otimes M(n,\C)$, with the notation $M(n,\C)$ denoting the $n\times n$ matrices with coefficients in $\C$.
Define $L_{jk}$ to be standard generators of the Lie algebra $\so(3)$ using the $n$-dimensional irreducible representation. The Grosse-Presnajder operator on $\mathcal H_{GP}$  is
\begin{equation}\label{GP}
\widetilde d=1+\sum_{j<k}\sigma^j\sigma^k\otimes [L_{jk},\cdot\,]
\end{equation}
One expects for a space of even dimension that there should be a chirality operator, whose eigenspaces determine the left- and right-handed modes.
 It is readily seen that the spectrum of $\widetilde d$ is not symmetric about 0, so there can be no chirality operator (although there is some literature proposing a degenerate version of a chirality operator \cite{Balachandran:1999qu}).

A new definition for the Dirac operator for a fuzzy sphere is made using four-component spinors.
 The Hilbert space is thus  $\mathcal H=\C^4\otimes M(n,\C)$. It admits a chirality operator and all of the axioms of a real spectral triple are obeyed. The Dirac operator can be decomposed into a direct sum of the Grosse-Presnajder operator with its negative, and so has a symmetric spectrum. One particular noteworthy feature of the fuzzy sphere is that it is necessary to use the type $(1,3)$ Clifford algebra to arrive at the correct algebraic properties for a sphere, in which the parameter called the signature in \cite{Barrett:2006qq}, or the KO-dimension in \cite{Connes:2006qv}, is equal to 2. The name KO-dimension will be used for this parameter, though it is not necessary to know anything about the K-theory from which the name is derived.

The construction of the fuzzy sphere is generalised to a Hilbert space based on non-square matrices in \S\ref{sec:gensphere} and is called the generalised fuzzy sphere. This is based on the operator introduced in  \cite{Grosse:1995jt}, but the Dirac operator presented here is new.

The relation of these to the commutative sphere $S^2$ is discussed in \S\ref{sec:commgeom}. The main feature is that the fuzzy sphere corresponds to two copies of the spinor field on $S^2$. This `fermion doubling' can be understood from the commutative geometry as an essential step to trivialise both the spinor bundle and the chirality operator before approximating by the fuzzy space. This principle, that a fuzzy space can only approximate trivial spinor data, seems an interesting idea that is worth testing in other examples. By comparing with the commutative Dirac operator, it is shown that the Grosse-Presnajder operator is in fact more analogous to what is called in the literature a modified Dirac operator, which is the commutative Dirac operator times the chirality operator.  

In the case of the generalised fuzzy sphere, the commutative model involves twisting the Dirac operator with a monopole bundle, doubling the fermions by taking two copies of the bundle, and then adding a constant matrix to the Dirac operator that is analogous to a mass mixing matrix in particle physics. Only then is the  commutative Dirac operator fully analogous to the generalised fuzzy case. 

One finds that the eigenvalues of the Dirac operator of the fuzzy sphere, and generalised fuzzy sphere, are given in terms of the Casimir operators of the Lie algebra $\so(3)$ by the same formula as in the commutative case. Thus the eigenvalues coincide whenever the representations of $\so(3)$ correspond. The difference in the fuzzy case is simply that only a finite list of irreducible representations occurs. This can be considered as the same representation-theoretic content of the spinor fields on $S^2$ but with the implementation of a short-wavelength cut-off.

Finally, some remarks about the motivation for this study. As remarked earlier, in quantum gravity one expects a cut-off in the eigenvalues of the Dirac operator at the Planck scale. However quantum gravity is not the study of single geometries but, from a functional integral point of view, the study of integration over a space of geometries. The interesting thing is that for a finite-dimensional Hilbert space, once the algebra of coordinates and Clifford algebra are fixed, the
set $\mathcal G$ of possible Dirac operators describing different geometries for this algebra is a finite-dimensional vector space of matrices. 

This suggests the definition of quantum geometries and random geometries using an integral over $\mathcal G$.
The translation-invariant measure $\dd\mathcal D$ is uniquely determined, up to an overall scaling. Given an action $S\colon\mathcal G\to\R$ and an `observable' function  $f\colon\mathcal G\to\R$, the expectation value of $f$ in the quantum geometry is defined as
\begin{equation}\langle f\rangle=\frac{\int f(D) e^{iS(D)} \;\dd\mathcal D}{\int e^{iS(D)} \;\dd\mathcal D}.
\end{equation}
Likewise, there is a `Euclidean' version in which the geometries are random in the probabilistic rather than quantum sense. Here, it is required that $S\ge0$, and
\begin{equation}\langle f\rangle=\frac{\int f(D) e^{-S(D)} \;\dd\mathcal D}{\int e^{-S(D)} \;\dd\mathcal D}.
\end{equation} Since the Dirac operator is a matrix, these integrals can be viewed as matrix models, with very particular constraints on the form of the matrices. It is hoped that the definitions and results in this paper will be useful in future studies of these notions of quantum geometry and random geometry.

\section{Gamma matrices}\label{sec:gamma}
The theory of gamma matrices is well-known \cite{BudinichTrautman}. However a variety of different conventions and constructions are used and not all the results needed are gathered in one place. So an explicit construction of the gamma matrices and the associated structures that are needed is given here. 

  The notation $x^*$ is used for the (Hermitian) adjoint of a linear operator on  a Hilbert space, and also, interchangably, for the Hermitian conjugate of a matrix. The complex conjugate of a complex number $c$ is written $\overline c$.

Let $\eta^{ab}$ be a diagonal matrix with diagonal entries $\pm1$. If the diagonal entries have $p$ occurences of $+1$ and $q$ of $-1$, then the matrix is of type $(p,q)$. The  dimension is $n=p+q$.   The matrix $\eta$ determines a real Clifford algebra over $\R^n$ and a representation of the Clifford algebra as complex matrices is called a Clifford module. The characteristics of Clifford modules depend to a large extent on $q-p$, with a pattern that repeats with period eight. The parameter $s\equiv q-p \mod 8$ is called the signature of the Clifford algebra.
 
Explicitly, a Clifford module for $\eta$ is a set of $n$ matrices $\gamma^1,\gamma^2,\ldots \gamma^n$ such that
\begin{equation}\gamma^a\gamma^b+\gamma^b\gamma^a=2\eta^{ab}.\end{equation}
This is said to be a Clifford module of type $(p,q)$ and the matrices are called gamma matrices.  

Denote the vector space these matrices act on as $V=\C^k$. It will be assumed that $V$ carries the standard Hermitian inner product $(v,w)=\sum_j \bar v_j w_j$ and that the $\gamma^a$ are all unitary. This assumption is not at all restrictive since the $\gamma^a$ generate a finite group and any finite group representation is unitary for some Hermitian inner product. All that is required here is to use a basis of $V$ in which this Hermitian form is the standard one. It follows that if $(\gamma^a)^2=1$, then $\gamma^a$ is a Hermitian matrix, and if $(\gamma^a)^2=-1$ then $\gamma^a$ is anti-Hermitian.  In the same way as for finite groups, any Clifford module decomposes into a direct sum of irreducible ones.

The order of the gamma matrices (up to even permutations of $1,2,\ldots,n$) is important. Using this order, one can form the product
\begin{equation}P=\gamma^1\gamma^2\ldots\gamma^n\end{equation}
that plays a key role in the theory. It is easily shown that
\begin{equation}P^2=(-1)^{\frac12 s(s+1)}\end{equation}
so that defining the chirality operator
\begin{equation}\label{gammadef}\gamma=i^{\frac12 s(s+1)}P,\end{equation}
one has $\gamma^2=1$, $\gamma^*=\gamma$.  (In the physics literature for $n=4$, $\gamma$ is called `$\gamma^5$'). For $n$ even, $\gamma$ anti-commutes with all gamma matrices. For $n$ odd it commutes with all gamma matrices.

In the case of irreducible gamma matrices, the dimension of the spinors is $k=2^{n/2}$ for $n$ even and $k={2^{(n-1)/2}}$ for $n$ odd. For $n$ even, the irreducible representation of the Clifford algebra is unique up to equivalence. For $n$ odd, there are two inequivalent representations characterised by $\gamma=1$ or $\gamma=-1$. 

These statements can be proved by realising that a pair of gamma matrices, say $\gamma^1$ and $\gamma^2$, generates a Clifford module, and an eigenvector of $\gamma^1\gamma^2$ generates an irreducible submodule of dimension 2. Thus for $n$ even, a simultaneous eigenvector of $\gamma^1\gamma^2, \gamma^3\gamma^4,\ldots,\gamma^{n-1}\gamma^n$ generates an irreducible Clifford module of dimension $2^{n/2}$ containing all of the possible simultaneous eigenvalues, hence the uniqueness. If $n$ is odd, one uses a simultaneous eigenvector of $\gamma^1\gamma^2, \gamma^3\gamma^4,\ldots,\gamma^{n-2}\gamma^{n-1},\gamma^n$. The action of the gamma matrices generates one half of the possible simultaneous eigenvectors, classified by the eigenvalue of $\gamma$. Thus there are two inequivalent irreducible Clifford modules.

\subsection{Real structures}
A real structure for a Clifford module is an antilinear operator $C\colon V\to V$ satisfying
\begin{itemize}\item $C^2=\epsilon=\pm 1$
\item $(Cv,Cw)=(w,v)$
\item $C\gamma^a=\epsilon'\gamma^a C$, with $\epsilon'=\pm1$.
\end{itemize}

The signs $\epsilon$ and $\epsilon'$ are given in figure \ref{signtable}.
\begin{figure}
$$\begin{tabular}{|l|rrrrrrrr|}
  \hline
  s &0&1&2&3&4&5&6&7\\
\hline
$\epsilon$&1&1&-1&-1&-1&-1&1&1\\
%\hline
$\epsilon'$&1&-1&1&1&1&-1&1&1\\
%\hline
$\epsilon''$&1&1&-1&1&1&1&-1&1\\
  \hline
&$\R$&$\C$&$\HH$&$\HH$&$\HH$&$\C$&$\R$&$\R$\\
\hline
\end{tabular}$$
\caption{Table of signs for Clifford modules and real spectral triples.}\label{signtable}
\end{figure}
One can calculate an additional sign $\epsilon''=\pm1$ such that
\begin{equation} C\gamma=\epsilon''\gamma C.\end{equation}
This sign is also given in figure \ref{signtable}. Since the gamma matrices generate a simple algebra over $\R$, this is a matrix algebra over one of the real division algebras $\R$, $\C$ or $\HH$. This is indicated in the final row of figure \ref{signtable}. %When $\epsilon'=1$, the division algebra is $\R$ or $\HH$, depending on $\epsilon$. For $s=1,5$, there is no antilinear operator that commutes with the algebra, so one uses the fact that $i\gamma^a$ determines a set of gamma matrices for $s'\equiv-s \mod 8$. Then one borrows the real structure for the $i\gamma^a$, which now anti-commutes with the gamma matrices $\gamma^a$, explaining the sign $\epsilon'$.

\begin{examples}\label{ex:smallcliff} Examples for $n\le2$ are
\begin{itemize}\item Type (0,0) 
\begin{equation}\label{type00}\begin{gathered}
s=0, V=\C^1, C\begin{pmatrix}v_1 \end{pmatrix}=\begin{pmatrix}\overline v_1 \end{pmatrix}\\
  \gamma=\begin{pmatrix}1\end{pmatrix}
\end{gathered}\end{equation}
\item Type (1,0)
\begin{equation}\label{type10}\begin{gathered}
s=7, V=\C^1, C\begin{pmatrix}v_1\end{pmatrix}=\begin{pmatrix}\overline v_1 \end{pmatrix}\\
\gamma^1=\begin{pmatrix}1\end{pmatrix}, \gamma=\begin{pmatrix}1\end{pmatrix}
\end{gathered}\end{equation}
\item Type (0,1)
\begin{equation}\label{type01}\begin{gathered}
s=1,  V=\C^1, C\begin{pmatrix}v_1\end{pmatrix}=\begin{pmatrix}\overline v_1 \end{pmatrix}\\
\gamma^1=\begin{pmatrix}-i\end{pmatrix}, \gamma=\begin{pmatrix}1\end{pmatrix}
\end{gathered}\end{equation}
\item Type (2,0)
\begin{equation}\label{type20}\begin{gathered}
s=6 ,  V=\C^2, C\begin{pmatrix}v_1\\v_2\end{pmatrix}=\begin{pmatrix}\overline v_1\\\overline v_2 \end{pmatrix}\\
\gamma^1=\begin{pmatrix}1&0\\0&-1\end{pmatrix},
\gamma^2=\begin{pmatrix}0&1\\1&0\end{pmatrix},
\gamma= \begin{pmatrix}0&i\\-i&0\end{pmatrix}  
\end{gathered}\end{equation}
\item Type (1,1)
\begin{equation}\label{type11}\begin{gathered}
s=0 ,  V=\C^2, C\begin{pmatrix}v_1\\v_2\end{pmatrix}=\begin{pmatrix}\overline v_1\\\overline v_2 \end{pmatrix}\\
\gamma^1=\begin{pmatrix}1&0\\0&-1\end{pmatrix}, 
\gamma^2=\begin{pmatrix}0&1\\-1&0\end{pmatrix},  
\gamma=\begin{pmatrix}0&1\\1&0\end{pmatrix}
\end{gathered}\end{equation}
\item Type (0,2)
\begin{equation}\label{type02}\begin{gathered}
s=2 ,  V=\C^2, C\begin{pmatrix}v_1\\v_2\end{pmatrix}=\begin{pmatrix}\overline v_2\\-\overline v_1 \end{pmatrix}\\
\gamma^1=\begin{pmatrix}i&0\\0&-i\end{pmatrix}, \gamma^2=\begin{pmatrix}0&1\\-1&0\end{pmatrix},  \gamma=\begin{pmatrix}0&1\\1&0\end{pmatrix}
\end{gathered}\end{equation}
\end{itemize}
\end{examples}

\subsection{Products of Clifford modules}
An explicit construction of gamma matrices for all $(p,q)$ can be given using the product of Clifford modules. Given a Clifford module $\mathcal M_1$ given by $\{\gamma_1^a\}$ for $(p_1,q_1)$, with even $s_1=q_1-p_1$ and chirality operator $\gamma_1$,  and another Clifford module $\mathcal M_2$ given by $\{\gamma_2^b\}$ for $(p_2,q_2)$ (with any value of  $s_2$), then the product $\mathcal M=\mathcal M_1\circ\mathcal M_2$ is defined by the matrices
\begin{equation}\gamma_1^1\otimes 1,\;\gamma_1^2\otimes 1,\;\ldots,\; \gamma_1^{n_1}\otimes 1,\;\gamma_1\otimes \gamma_2^1,\;\gamma_1\otimes \gamma_2^2,\;\dots,\;\gamma_1\otimes \gamma_2^{n_2}
\end{equation}
acting on the space $V_1\otimes V_2$. This is a Clifford module of type $(p,q)=(p_1+p_2,q_1+q_2)$. The product $\mathcal M$ is irreducible if $\mathcal M_1$ and $\mathcal M_2$ are.

A calculation shows that the chirality operator $\gamma$ for the product is determined in terms of the chirality operators $\gamma_1$ and $\gamma_2$ by
\begin{equation}\gamma=
\begin{cases}\gamma_1\otimes \gamma_2 &(s_2 \text{ even})\\
\epsilon_1''\otimes \gamma_2&(s_2 \text{ odd}).
\end{cases}\end{equation}
The Hermitian inner product is defined by 
\begin{equation}(v_1\otimes w_1,v_2\otimes w_2)=(v_1,v_2)(w_1,w_2).\end{equation}
A real structure $C$ for the product is defined by
\begin{equation}C=\begin{cases}
C_1\otimes C_2 &(s_2 \text{ even}, \epsilon_1''=1)\\ 
C_1\otimes C_2\gamma_2 &(s_2 \text{ even}, \epsilon_1''=-1)\\ 
C_1\otimes C_2 &(s_2 \text{ odd}, \epsilon'=1)\\ 
C_1\gamma_1\otimes C_2 &(s_2 \text{ odd}, \epsilon'=-1)\\ 
\end{cases}\end{equation}
A calculation shows that this product operation for both the gamma matrices and the real structure is strictly associative (using the standard identification of vectors or matrices $(u\otimes v)\otimes w=u\otimes (v\otimes w)$), so it is not necessary to use parentheses for multiple products. 

Gamma matrices can be constructed for any $(p,q)$ by taking the product of a number of copies of the $n=2$ case, with the product on the right with an $n=1$  case if necessary. Using this construction, the signs in figure \ref{signtable} can be checked explicitly.

\begin{example}\label{ex:addgamma} If $ \mathcal M_2$ is the type $(1,0)$ module \eqref{type10}, the product amounts to adjoining the chirality operator to the end of the list of generators of $\mathcal M_1$, i.e., $\gamma^{n}=\gamma_1$. (This explains why the chirality operator for $n=4$ was originally called $\gamma^5$.) The new chirality operator is the scalar $\epsilon'$.

Similarly, if  $ \mathcal M_2$ is the type $(0,1)$ module \eqref{type01}, then $\gamma^{n}=-i\gamma_1$, and the new chirality operator is $-\epsilon'$. 

In both cases, $C=C_1$ for $\epsilon'=1$ and $C=C_1\gamma_1$ for $\epsilon'=-1$.
\end{example}

\begin{example}\label{ex:trivialproduct} If $\mathcal M_1$ is the type $(0,0)$ module  \eqref{type00}, then $ \mathcal M_1\circ\mathcal M_2=\mathcal M_2$. This shows that the structure of $\mathcal M_1$ is correct despite the chirality operator having no eigenvalue $-1$.\end{example}

There is also a product operation $\mathcal M_1\bullet \mathcal M_2$ defined when $s_1$ and $s_2$ are both odd. For this, $\mathcal M_1$ must be determined by a larger Clifford module $\widetilde{\mathcal M}_1$ of type $(p_1+1,q_1)$ by ignoring an additional generator $t$, with $t^2=1$. The generators for $\widetilde{\mathcal M}_1$ are
\begin{equation}t, \gamma_1^1,\gamma_1^2,\ldots,\gamma_1^{n_1}
\end{equation}
in that order. The second Clifford module $\mathcal M_2$ is $\{\gamma_2^1,\ldots,\gamma_2^{n_2}\}$ of type $(p_2,q_2)$ as before. Then the product $\widetilde{\mathcal M}_1\circ\mathcal M_2$ is a Clifford module of type $(p_1+p_2+1,q_1+q_2)$. Ignoring the generator $t\otimes 1$ gives the definition of $\mathcal M_1\bullet \mathcal M_2$. The generators are expressed using the chirality operator $\widetilde\gamma_1$ for  $\widetilde{\mathcal M}_1$
\begin{equation}\gamma_1^1\otimes 1,\;\gamma_1^2\otimes 1,\;\ldots,\; \gamma_1^{n_1}\otimes 1,\;
\widetilde\gamma_1\otimes \gamma_2^1,\;\widetilde\gamma_1\otimes \gamma_2^2,\;\dots,\;\widetilde\gamma_1\otimes \gamma_2^{n_2}.
\end{equation}
This is a Clifford module of type $(p,q)=(p_1+p_2,q_1+q_2)$. For irreducible $\widetilde{\mathcal M}_1$ and $\mathcal M_2$, the dimension of $V$, the space these matrices act in, is $2^{(n_1+1)/2}2^{(n_2-1)/2}=2^{(n_1+n_2)/2}$, and so the module is irreducible. 

The matrix $\widetilde\gamma_1$ can be written in terms of the chirality operator $\gamma_1$ for $\mathcal M_1$ and $t$,
\begin{equation}\widetilde\gamma_1=i^{s_1(s_1-1)/2}tP_1=i\epsilon_1't\gamma_1.\end{equation}

A real structure for $\mathcal M_1\bullet \mathcal M_2$  can be defined in terms of the real structure $\widetilde C_1$ for $\widetilde{\mathcal M}_1$  and the real structure $C_2$ for $\mathcal M_2$ as 
\begin{equation}C=\begin{cases}
\widetilde C_1\otimes C_2 &(\epsilon''=1)\\ 
\widetilde C_1 \gamma_1\otimes C_2\gamma_2 &(\epsilon''=-1)\\ 
\end{cases}\end{equation}

The examples  for $n_1=1$ are particularly useful. 
\begin{example} The $\bullet$ product for $(p_1,q_1)=(0,1)$ is constructed using representation \eqref{type11}. Then $s=s_2+1$. The gamma matrices, written in block matrix form, are
\begin{equation}\begin{pmatrix}0&1\\-1&0\end{pmatrix},\quad \begin{pmatrix}0&\gamma_2^a\\\gamma_2^a&0\end{pmatrix}, \quad a=1,\ldots, n_2
\end{equation}
The chirality operator is 
\begin{equation}\gamma=\epsilon''\begin{pmatrix}\gamma_2&0\\0&-\gamma_2\end{pmatrix}
\end{equation}
and the real structure is
\begin{equation}C=\begin{cases}\begin{pmatrix}C_2&0\\0&C_2\end{pmatrix}&(\epsilon''=1)\\
\begin{pmatrix}0&C_2\gamma_2\\-C_2\gamma_2&0\end{pmatrix}&(\epsilon''=-1).\end{cases}\end{equation}
\end{example}

\begin{example} \label{ex:product10} The $\bullet$ product for $(p_1,q_1)=(1,0)$ is constructed using representation \eqref{type20}. Then $s=s_2-1$. The gamma matrices, written in block matrix form, are
\begin{equation}\begin{pmatrix}0&1\\1&0\end{pmatrix},\quad \begin{pmatrix}0&i\gamma_2^a\\-i\gamma_2^a&0\end{pmatrix}, \quad a=1,\ldots, n_2
\end{equation}
The chirality operator is 
\begin{equation}\gamma=-\epsilon''\begin{pmatrix}\gamma_2&0\\0&-\gamma_2\end{pmatrix}
\end{equation}
and the real structure is
\begin{equation}C=\begin{cases}\begin{pmatrix}C_2&0\\0&C_2\end{pmatrix}&(\epsilon''=1)\\
\begin{pmatrix}0&C_2\gamma_2\\C_2\gamma_2&0\end{pmatrix}&(\epsilon''=-1).\end{cases}\end{equation}
\end{example}

\section{Finite spectral triples}\label{sec:spectral}

\subsection{Axioms for a real spectral triple}\label{sec:axioms}
\begin{definition}
A finite real spectral triple consists of
\begin{enumerate}\item A integer $s$ defined modulo 8, called the KO-dimension.\label{ax:ko}
\item A finite-dimensional Hilbert space $\mathcal H$, with Hermitian inner product $\langle\cdot,\cdot\rangle$. 
\item An algebra $\mathcal A$ over $\R$ with an anti-automorphism $*\colon \mathcal A\to\mathcal A$  that is involutory (i.e., $(ab)^*=b^*a^*$ and $a^{**}=a$ for all $a,b\in\mathcal A$).  This is the non-commutative generalisation of the algebra of coordinates.
\item A faithful action  $\rho\colon\mathcal A\to\End(\mathcal H)$ such that $\rho(a^*)=\rho(a)^*$, the adjoint of $\rho(a)$, for all $a\in\mathcal A$.
\item An operator $\Gamma\colon \mathcal H\to \mathcal H$ satisfying $\Gamma^*=\Gamma$, $\Gamma^2=1$. This is called the chirality operator.\label{ax:chiral1}
\item $\Gamma \rho(a)=\rho(a)\Gamma$ for all $a\in\mathcal A$. \label{ax:chiral2}
\item An anti-linear operator $J\colon\mathcal H\to \mathcal H$ that is unitary, $\langle Ju,Jv\rangle=\langle v,u\rangle$ for all $u,v\in\mathcal H$. This operator is called the real structure.
\item $J^2=\epsilon$, $J\Gamma=\epsilon''\Gamma J$, using the signs in figure \ref{signtable}.\label{ax:real}
\item $[\rho(a),J\rho(b)J^{-1}]=0$ for all $a,b\in\mathcal A$.\label{ax:bimodule}
\item An operator $D\colon \mathcal H\to \mathcal H$ satisfying $D=D^*$, called the Dirac operator.\label{ax:dirac}
\item $D\Gamma=-\Gamma D$ if $s$ is even; $D\Gamma=\Gamma D$ if $s$ is odd.\label{ax:Dgamma}
\item $JD=\epsilon'DJ$.\label{ax:epsilonprime}
\item $[[D,\rho(a)],J\rho(b)J^{-1}]=0$ for all $a,b\in\mathcal A$.\label{ax:firstorder}
\end{enumerate}
\end{definition}
These axioms are a subset of the axioms given in \cite{Connes:1996gi} and are satisfied by a Dirac operator on a (commutative) manifold. The axioms coincide with those given in  \cite{Chamseddine:2006ep} for a finite-dimensional $\mathcal H$, except that the stronger condition $\Gamma=1$ for $s$ odd is used there. This is not a significant difference. 

The Hilbert space space splits into the $\pm1$ eigenspaces of $\Gamma$ as $\mathcal H=\mathcal H_+\oplus\mathcal H_-$. In the case where $s$ is odd, the operator $\Gamma$ commutes with all the operators defining the real spectral triple, and so $\mathcal H_+$ and $\mathcal H_-$ are also spectral triples on which $\Gamma=1$ and $\Gamma=-1$ respectively.

\begin{definition}
The above data without a Dirac operator, i.e., $(s, \mathcal H,\mathcal A,\Gamma, J)$ satisfying axioms \ref{ax:ko}--\ref{ax:bimodule}, is called a fermion space. The set of all Dirac operators $\mathcal G=\{D\}$ is called the space of geometries for the fermion space.
\end{definition}

The space of geometries $\mathcal G$ is a vector space over $\R$, as follows from the linearity of axioms \ref{ax:dirac}--\ref{ax:firstorder}. For $s=3\text{ or } 7$, any multiple of the identity operator is always a Dirac operator.

The axioms imply that $\mathcal H$ has a left and a right action of $\mathcal A$ that commute with each other, making $\mathcal H$ a bimodule over $\mathcal A$. The actions are
\begin{equation}a\triangleright\psi\triangleleft b=\rho(a)J\rho(b)^*J^{-1}\psi.
\end{equation}

\subsection{The index}
For $s$ even, a fermion space has an index, defined as
\begin{equation}I=\tr\Gamma=\dim\mathcal H_+-\dim\mathcal H_-.\end{equation}
For the cases $s=2,6$ if $\psi$ is an eigenvector of $\Gamma$ then  $J\psi$ is an eigenvector with the opposite eigenvalue and so $I=0$. However this is argument does not hold for $s=0,4$. In fact for these cases, $\mathcal H_+$ and $\mathcal H_-$ are both fermion spaces, and clearly have non-zero index (unless the dimension is zero).

 The Dirac operator can be regarded as a pair of operators $D_+\colon \mathcal H_+\to\mathcal H_-$ and $D_-\colon \mathcal H_-\to\mathcal H_+$, with $D_-=D_+^*$. The kernel $\ker A$ of an operator $A$ is the eigenspace for eigenvalue $0$. The kernel of $D$ is an invariant subspace for $\Gamma$ and the index of the Dirac operator is defined as
\begin{equation}I_D=\tr_{\ker D}\Gamma=\dim\ker D_+-\dim\ker D_-.\end{equation}
In fact, this index depends only on the fermion space.
\begin{lemma}  $I_D=I$.
\end{lemma}
\begin{proof}The Hilbert space can be decomposed into the subspaces where the eigenvalue of $D$ is negative, zero or positive, $\mathcal H=\mathcal H^<\oplus\ker D\oplus\mathcal H^>$. Since $\Gamma$ maps $\mathcal H^>\to\mathcal H^<$ and  $\mathcal H^<\to\mathcal H^>$, its trace on $\mathcal H^<\oplus\mathcal H^>$ is zero.
\end{proof}

\subsection{The Frobenius form}

Since $\mathcal A$ is a $*$-algebra of matrices, it follows that it is semisimple, and is thus isomorphic to a direct sum of simple matrix algebras, $\mathcal A\cong\oplus_i \mathcal A_i$. According to the Artin-Wedderburn theorem, each simple summand is isomorphic to an algebra of $n\times n$ matrices over a division ring, $\mathcal A_i\cong M(n_i,\mathbb D_i)$, with $\mathbb D_i=\R$, $\C$ or $\HH$, regarded as an algebra over $\R$. 

Since the algebra $\mathcal A$ is semisimple, it can be given the structure of a Frobenius algebra. This is a choice of 
 a linear form $\phi\colon\mathcal A\to \R$ such that the bilinear form $(a,b)\mapsto\phi(ab)$ is non-degenerate. The form is called a Frobenius form, and an algebra with a Frobenius form is called a Frobenius algebra.

The matrix trace on $M(n,\mathbb D)$ is denoted by $\tr$ and gives an element of $\mathbb D$. This can be used to construct Frobenius forms.  In fact there is a canonical choice of a Frobenius form. Writing $a=\oplus a_i$ in the decomposition into matrix algebras, the canonical form is
\begin{equation}\phi_0(a)=\sum_in_i\dim \mathbb D_i\;\Re\tr(a_i),
\end{equation}
using the notation $\dim \mathbb D$ for the dimension of the real algebra $\mathbb D$, so that $\dim \R=1$,  $\dim \C=2$, $\dim \HH=4$. A standard result \cite{KochBook} is that any other Frobenius form can be written 
\begin{equation}\label{frobenius}\phi(a)=\phi_0(xa)\end{equation}
 for some invertible element $x\in\mathcal A$.

The bilinear form has an inverse,  $B=\sum B_L\otimes B_R\in\mathcal A\otimes \mathcal A$ defined by
\begin{equation}\label{eq:Bdef}\sum \phi(a B_L)B_R=a=\sum B_L\phi(B_Ra)
\end{equation}
for all $a\in\mathcal A$. This has the property 
\begin{equation}\label{eq:Bcentral}\sum a B_L\otimes B_R=\sum  B_L\otimes B_R\,a
\end{equation}
This is proved by noting that $ \sum \phi(ba B_L)\otimes B_R=ba=\sum  \phi(b B_L)\otimes B_R\,a$ for all $b\in\mathcal A$.
Applying the algebra multiplication to $B$ results in $z=\sum  B_L B_R$ which is a central element of $\mathcal A$ according to \eqref{eq:Bcentral}. If $z=1$, $\phi$ is called a special Frobenius form. 

Some further properties a Frobenius form might have are $*$-invariance 
\begin{equation}\phi(a^*)=\phi(a)\end{equation}
and symmetry
\begin{equation}\phi(ab)=\phi(ba).\end{equation}

The following is proved in \cite{Barrett:2013kza}.
\begin{lemma} The canonical form $\phi_0$ is a Frobenius form that is special, $*$-invariant 
 and symmetric. 
\end{lemma}
By picking suitable $x$, then using \eqref{frobenius} one can generate all Frobenius forms with a subset of these three properties. If $x=x^*$, the form is $*$-invariant. If $x$ is central then the form is symmetric. If $\tr(x_i^{-1})=n_i$ for each $i$, then the form is special.

\subsection{The Dirac operator}
 For finite spectral triples, the structure of a Dirac operator is given by the following lemma, generalising the result given in \cite{Krajewski:1996se}. The lemma uses a choice of a special *-invariant Frobenius form.

\begin{lemma}\label{lem:dirac} A Dirac operator for a finite spectral triple has the form 
\begin{equation}\label{eq:dirac}D=\theta+\epsilon'J\theta J^{-1}.\end{equation}
with $\theta$ an operator on $\mathcal H$ satisfying 
\begin{itemize}\item $\theta^*=\theta$ 
\item $[\theta,J\rho(a)J^{-1}]=0$ for all $a\in\mathcal A$
\item $\theta\Gamma+\Gamma\theta=0$ ($s$ even), $\theta\Gamma-\Gamma\theta=0$ ($s$ odd) .
\end{itemize}
 Conversely, given a fermion space, any such operator $\theta$ defines a Dirac operator by \eqref{eq:dirac}.\end{lemma}
Note that the second condition means that $\theta$ commutes with the right action of $\mathcal A$. Then $J\theta J^{-1}$ commutes with the left action.
\begin{proof} Let $D$ be a Dirac operator for a finite spectral triple.

Let $\phi$ be a special Frobenius form for $\mathcal A$, such that $\phi(a)=\phi(a^*)$. Denote the inverse of $\phi$ by $B$. If $K\colon\mathcal H\to\mathcal H$ is an operator, then define operators 
\begin{equation}\pi K=\sum  \rho(B_L)K \rho(B_R)\end{equation}
and
\begin{equation}\pi' K=J\pi(J^{-1}KJ)J^{-1}=\sum J \rho(B_L)J^{-1}KJ \rho(B_R)J^{-1}.\end{equation}
Due to \eqref{eq:Bcentral}, $\pi K$ commutes with the left action and  $\pi' K$ commutes with the right action, $[\pi K,\rho(a)]=[\pi'K,J\rho(a)J^{-1}]=0$ for all $a\in\mathcal A$. Then $\pi^2K=\pi K, \pi'^2K=\pi'K, \pi'\pi K=\pi\pi'K$.
From \eqref{eq:Bdef}, $a^*=\sum B_R^*\phi(B_L^*a^*)$ for all $a^*$, so that $B=\sum B_R^*\otimes B_L^*$.
 Hence 
\begin{equation}(\pi K)^*=\sum  \rho(B_R^*)K^* \rho(B_L^*)=\pi(K^*)\end{equation}

The first-order condition, axiom \ref{ax:firstorder}, implies that $\pi'[D,\rho(a)]=[D,\rho(a)]$ for all $a\in\mathcal A$. Hence $[\pi'D-D,\rho(a)]=0$, and this implies $\pi(\pi'D-D)=\pi'D-D$, which re-arranges to give a decomposition
\begin{equation}D=\pi D+\pi' D-\pi\pi'D.
\end{equation}
Define
\begin{equation}\label{eq:theta}\theta=\pi'D-\frac12\pi\pi'D
\end{equation}
Then $\epsilon'J\theta J^{-1}=\pi D-\frac12\pi'\pi D$, which proves \eqref{eq:dirac}. 

For the properties of $\theta$, note that for any $K$, $(JKJ^{-1})^*=JK^*J^{-1}$. Therefore $(\pi' K)^*=\pi'(K^*)$ and so $\theta^*=\theta$. The property $[\theta,J\rho(a)J^{-1}]=0$ follows from \eqref{eq:Bcentral} and the commutation with $\Gamma$ follows from axioms \ref{ax:chiral2}, \ref{ax:real} and \ref{ax:Dgamma}.

For the converse, suppose $\theta$ is an operator satisfying the three given conditions. Then defining $D$ by  \eqref{eq:dirac}, axioms \ref{ax:dirac}--\ref{ax:firstorder} are easily verified.
\end{proof}

The decomposition of the Dirac operator is into the part $\theta$ that commutes with the right action and the part $J\theta J^{-1}$ that commutes with the left action. Therefore one expects that if the Dirac operator is allowed a part that commutes with both left and right actions (a bimodule map), then the decomposition may not be unique. % There is of course a canonical decomposition given by the canonical Frobenius form, but a different Frobenius form may give a different decomposition.

\begin{lemma}\label{lem:ambiguity}
Suppose $\theta_1$ and $\theta_2$ satisfy the conditions of lemma \ref{lem:dirac} for a fixed Dirac operator $D$. Then $\psi=\theta_1-\theta_2$ satisfies
\begin{itemize}\item $\psi^*=\psi$ 
\item $[\psi,\rho(a)]=0$  for all $a\in\mathcal A$
\item $[\psi,J\rho(a)J^{-1}]=0$ for all $a\in\mathcal A$
\item $\psi\Gamma+\Gamma\psi=0$ ($s$ even), $\psi\Gamma-\Gamma\psi=0$ ($s$ odd) .
\item $\psi+\epsilon'J\psi J^{-1}=0$
\end{itemize}
\end{lemma}

In particular, $\psi$ is a bimodule map. For some fermion spaces there are no non-zero bimodule maps respecting the chirality operator and so the decomposition is unique. This is the situation in \cite{Krajewski:1996se}, where an additional axiom (orientability) rules out non-zero bimodule maps. However it was noted in \cite{Stephan:2006em} that there are interesting examples that allow non-zero bimodule maps and so the orientability axiom has been dropped in the more recent literature.

The role of the Frobenius form in determining a unique $\theta$ is as follows. Given a Dirac operator $D$, suppose  that $\theta_2=\theta$ is determined by \eqref{eq:theta} using a Frobenius form with projection operators $\pi$, $\pi'$. Then for any $\theta_1$ that also gives $D$ via \eqref{eq:dirac},
\begin{equation}\label{eq:psifrob}\psi=\theta_1-\theta=\frac12\left(\pi\theta_1-\epsilon'J(\pi\theta_1)J^{-1}\right),\end{equation}
for which the proof is below.  
It follows that $\theta$ is the unique operator satisfying  the conditions of lemma \ref{lem:dirac} and also
\begin{equation}\pi\theta=\epsilon'J(\pi\theta)J^{-1}.\end{equation}

\begin{proof}[Proof of \eqref{eq:psifrob}]
Write $D=\theta_1+\epsilon'J\theta_1 J^{-1}$. Then
\begin{equation} \pi'D=\theta_1+\epsilon'\pi'(J\theta_1 J^{-1})=\theta_1+\epsilon' J(\pi\theta_1) J^{-1}
\end{equation}
\begin{equation} \pi \pi'D=\pi\theta_1+\epsilon' J(\pi\theta_1) J^{-1}
\end{equation}
which leads directly to  \eqref{eq:psifrob}.
\end{proof}

\subsection{Transformations}\label{sec:transf}
A transformation is a unitary operator $U\colon\mathcal H\to\mathcal H$ 
that respects the structures of the fermion space
\begin{itemize}
\item $U\rho(a)U^{-1}=\rho(b)$ for some $b\in \mathcal A$, for all $a\in\mathcal A$
\item $U\Gamma=\Gamma U$
\item $UJ=JU$
\end{itemize}
Note that the first condition implies that $U$ determines an automorphism of the algebra $\mathcal A$.

A transformation takes a Dirac operator $D$ to another Dirac operator $D'$ with the same spectrum by the formula
\begin{equation}\label{eq:transD} D'=U D U^{-1}\end{equation}
This determines a linear map on the space of Dirac operators $\mathcal G$. A symmetry is a transformation $U$ such that $D'=D$. 

An example of a transformation is given by a unitary element $u\in\mathcal A$. Then the transformation is $U=\rho(u)J\rho(u) J^{-1}$. This example has the interpretation as a gauge transformation in the example of the standard model spectral triple.

 For  $s$ even, the chiral rotation operator is
\begin{equation} R=e^{-i\pi\Gamma/4}.\end{equation}
This operator obeys $R^2=-i\Gamma$.  The modified Dirac operator $\widetilde D$ is defined as
\begin{equation}\label{eq:modified}\widetilde D=R D R^{-1}=R^2D=-i\Gamma D.\end{equation}
For $s=2$ or $6$ (the cases where $\epsilon''=-1$) this is a transformation, since $RJ=JR$, and $\widetilde D$ is also a Dirac operator. For $s=0$ or $s=4$, one has $RJ=R^{-1}J$ and so this is not a transformation of the fermion space and  $\widetilde D$ is not a Dirac operator.

The structure of transformations for fuzzy spaces is considered further in  \S\ref{sec:fuzzytrans}.

\section{Matrix geometries}\label{sec:matrixgeometry}

A matrix geometry is a spectral triple with a fermion space $(s, \mathcal H,\mathcal A,\Gamma, J)$ that is a product of a type $(0,0)$ matrix geometry with a Clifford module. Thus the type $(0,0)$ case is defined first.

\begin{definition} A type $(0,0)$ matrix geometry is a real spectral triple with KO-dimension $s_0=0$ and chirality operator $\Gamma_0=1$, together with the only possible Dirac operator, $D_0=0$, and a real structure $J_0$. The Hilbert space $\mathcal H_0$ decomposes into a direct sum of irreducible bimodules over the algebra $\mathcal A_0$. It is required that these irreducible bimodules are all inequivalent. 
\end{definition}

The name `matrix geometry' arises from the fact that examples of this definition are given by a matrix construction. In fact, this construction turns out to be the general case, up to isomorphism.

\begin{example}\label{ex:matrixgeom}
Let $\C^n$, with its standard Hermitian inner product, be a faithful left module for $*$-algebra $\mathcal A_0$ such that the irreducible sub-modules are all inequivalent to each other. Let $\mathcal H_0$ be a $\C$-linear vector subspace of $M(n,\C)$
such that matrix product $am\in \mathcal H_0$ and $m^*\in \mathcal H_0$ for all $a\in\mathcal A_0$ and $m\in \mathcal H_0$. Then the representation of $\mathcal A_0$ in $\mathcal H_0$ is just the matrix multiplication, $\rho_0(a)m=am$.

The Hermitian inner product on $\mathcal H_0$ given by $\langle m,m'\rangle=\tr m^*m'$ makes $\mathcal H_0$ a Hilbert space.  The real structure on $\mathcal H_0$ is given by Hermitian conjugation, $J_0=*$. 
The action of $a^*$ can be computed using the inner product
\begin{equation}
\langle  m,a m'\rangle
=\tr m^*am'
=\tr ((a^*m)^*m')=\langle a^* m, m'\rangle
\end{equation}
so that $a^*$ indeed acts in $\mathcal H_0$ as the adjoint of $a$. The unitarity of $J_0$ follows from 
\begin{equation}\langle J_0m,J_0m'\rangle=\tr mm'^*=\tr m'^*m=\langle m',m\rangle\end{equation}
The right action is $J_0^{\vphantom{-1}}a^*J_0^{-1}m=J_0a^*m^*=ma$, which is matrix multiplication on the right, and commutes with the left action. 
Thus this construction gives a type $(0,0)$ matrix geometry. 
\end{example}

\begin{lemma} A type $(0,0)$ matrix geometry is isomorphic to one constructed in example \ref{ex:matrixgeom}.
\end{lemma}

\begin{proof} Starting with a type $(0,0)$ matrix geometry $(\mathcal H_0,\mathcal A_0,\mathcal J_0)$, construct  a left action of $\mathcal A_0$ on $\C^n$, for some $n$, that is the direct sum of irreducible left $\mathcal A_0$-modules containing one representative from each equivalence class. 
Then $M(n,\C)$ is a bimodule over $\mathcal A_0$ and $\mathcal H_0'\subset M(n,\C)$ is defined as the linear subspace that is isomorphic to $\mathcal H_0$ as a bimodule. The matrix construction is completed with $J_0'=*$. 

Denote the unitary isomorphism of bimodules $\phi\colon\mathcal H_0'\to\mathcal H_0$. Then $u=J_0'^{-1}\phi^{-1}J_0^{\vphantom{-1}}\phi$ is a unitary bimodule map $\mathcal H_0'\to\mathcal H_0'$ satisfying $u^*=J_0'uJ_0'^{-1}$. The main point of the proof is to show that $u$ has a square root, i.e. $u=r^2$, with unitary bimodule map $r\colon\mathcal H_0'\to\mathcal H_0'$ satisfying $r^*=J_0'rJ_0'^{-1}$. Assuming this is true, then defining $\psi=\phi r^{-1}$, it follows that
\begin{equation} J_0'^{-1}\psi^{-1}J_0^{\vphantom{-1}}\psi= J_0'^{-1}r\phi^{-1}J_0^{\vphantom{-1}}\phi r^{-1}=r^{-1}ur^{-1}=1.\end{equation}
Thus $\psi$ is the required isomorphism between $\mathcal H_0'$ and $\mathcal H_0$.

The existence of $r$ is due to the fact that $u$ is a scalar on each irreducible sub-bimodule. In the decomposition into irreducibles $\mathcal H_0'=\oplus_{ij}\mathcal H'_{ij}$, with $i$ labelling the irreducibles for the left action of $\mathcal A$ and $j$ for the right action,  $u=\oplus u_{ij}1_{ij}$ with $1_{ij}$ the identity matrix, $u_{ij}\in\C$, $|u_{ij}|=1$. Then $u^*=\oplus \overline u_{ij}1_{ij}$ and  $J_0'uJ_0'^{-1}=\oplus \overline u_{ji}1_{ij}$. Hence  $u_{ji}=u_{ij}$ and it suffices to choose the square roots $r_{ij}=u_{ij}^{1/2}$ so that also $r_{ji}=r_{ij}$.
\end{proof}

The simplest examples are as follows.

\begin{example}\label{ex:diagonal}
Let $a\in \mathcal A_0=M(n,\C)$ (considered as a real algebra) and $m\in\mathcal H_0= M(n,\C)$.
The bimodule action is 
\begin{equation} a\triangleright m\triangleleft b
=amb,
\end{equation}
and the real structure
\begin{equation}J m=m^*\end{equation}
\end{example}

\begin{example}\label{ex:offdiagonal}
Let $a\in \mathcal A_0=M(n_1,\C)\oplus M(n_2,\C)$ and $m\in\mathcal H_0\subset M(n_1+n_2,\C)$, represented as block matrices as follows
\begin{equation}a=\begin{pmatrix}a_1& .\\ .&a_2\end{pmatrix},\quad m=\begin{pmatrix}.&m_1\\m_2& .\end{pmatrix}.\end{equation}
The bimodule action is 
\begin{equation} a\triangleright m\triangleleft b
=\begin{pmatrix}.&a_1m_1b_2\\a_2m_2b_1&.\end{pmatrix},
\end{equation}
and the real structure
\begin{equation}J m=\begin{pmatrix}.&m^*_2\\m_1^*&.\end{pmatrix}.\end{equation}
Note that this example  for $n_1\ne n_2$ does not have a `separating vector' in the terminology of \cite{Chamseddine:2007hz}.
\end{example}

\begin{example}\label{ex:complexconj}
This example is constructed from example \ref{ex:offdiagonal} with $n_1=n_2=n$ but with $\mathcal A_0$ the subalgebra of $M(n,\C)\oplus M(n,\C)$ given by $a_2=\overline a_1$, the complex conjugate matrix.
\end{example}

A Clifford module $V$ of type $(p,q)$ also determines a fermion space. The KO-dimension is $s\equiv q-p \mod 8$, the chirality $\gamma$ and real structure $C$ are determined by the Clifford module, and the algebra of coordinates is $\R$. The vector space $V$ is called the space of spinors, and the module will be assumed to be irreducible in the sense given in the following definition.

\begin{definition} \label{def:matrixgeom} A type $(p,q)$ matrix geometry is a real spectral triple\\ $(s,\mathcal H,\mathcal A,\Gamma,J,D)$ with fermion space the tensor product of a Clifford module of type $(p,q)$ and a type $(0,0)$ matrix geometry. The Clifford module is required to be irreducible if $s$ is even, and if $s$ is odd, the eigenspaces $V_\pm\subset V$ of the chirality operator $\gamma$ are required to be irreducible.
Explicitly, 
\begin{itemize}\item $s\equiv q-p \mod 8$
\item $\mathcal H=V\otimes\mathcal H_0$
\item $\langle v\otimes m,v'\otimes m'\rangle=(v,v')\langle m,m'\rangle$
\item $\mathcal A=\mathcal A_0$
\item $\rho(a)(v\otimes m)=v\otimes (\rho_0(a)m)$
\item $\Gamma=\gamma\otimes 1$
\item $J=C\otimes J_0$
\end{itemize}

Note that in the odd case, either of $V_\pm$ may be zero-dimensional, in which case $V$ is itself irreducible.
\end{definition}
% $\gamma^a$, $a=1,\ldots,n$, with $n=p+q$. The space these act in is $V$, as in \S\ref{sec:gamma}
%and this satisfies axioms \ref{ax:chiral1}, \ref{ax:chiral2}.
% The real structure  is  defined in terms of the real structure $C$ for the Clifford module,
%$$J(v\otimes m)=(Cv)\otimes m^*,$$
%which is unitary. Axiom \ref{ax:real} follows from the corresponding properties of $C$, using the signs $\epsilon$, $\epsilon''$ determined by $s$.

\subsection{Matrix Dirac operator}
A Dirac operator in a matrix geometry is any operator allowed by the axioms  \ref{ax:dirac}--\ref{ax:firstorder}. According to lemma \ref{lem:dirac}, this is determined by $\theta\in\End(\mathcal H)\cong \End(V)\otimes\End(\mathcal H_0)$, and can be written
\begin{equation}\theta=\sum_i\omega^i\otimes X_i.
\end{equation}
with the $\omega^i$ a linearly independent set. The condition $[\theta,J\rho(a)J^{-1}]=0$ for all $a\in\mathcal A$ is then equivalent to requiring that 
\begin{equation}[X_i,J_0^{\vphantom{-1}}\rho_0(a)J_0^{-1}]=0\quad \text{ for all } a\in\mathcal A.
\end{equation}

Let $\Omega$ be the $\R$-algebra generated by the matrices $\gamma^a$ in the Clifford module, and $\Omega^+$ the subalgebra generated by even products of the gamma matrices.
 The irreducibility assumptions for a fuzzy space imply that $\omega^i$ can be chosen to be elements of $\Omega$, and this will be assumed in the following. To examine this in more detail, the cases of $s$ odd and $s$ even must be looked at separately.

For $s$ even, $V$ is irreducible and so $\Omega\otimes_{\R}\C=\End(V)$. The algebra is graded by $\Omega=\Omega^+\oplus\Omega^-$, with $\Omega^-$ the vector space generated by products of odd numbers of $\gamma^a$.
The condition  $\theta\Gamma+\Gamma\theta=0$ is satisfied if the $\omega^i$ are taken to be a basis of $\Omega^-\otimes_{\R}\C$. Finally, any complex coefficients can be absorbed into the $X_i$, so that  it can be assumed that $\omega^i\in\Omega^-$.

For $s$ odd, $\Omega\otimes_{\R}\C=\End(V^+)\oplus\End(V^-)$. However, the $\omega^i$ also commute with $\gamma$ and so can be chosen as a basis of $\Omega\otimes_{\R}\C$ with $\omega^i\in\Omega$.

The final condition from  lemma \ref{lem:dirac} is $\theta=\theta^*$. Since $\theta^*=\sum_i\omega^{i*}\otimes X_i^*$, the condition is satisfied if $\omega^i$ and $X_i$ are either both Hermitian or both anti-Hermitian. Such a basis of $\Omega^-$ or $\Omega$ exists by taking products of $\gamma^i$, as will become clear from the examples shortly. 
%Therefore write
%\begin{equation}\theta=\sum_j\alpha^j\otimes L_j +\sum_k \tau^k\otimes H_k\end{equation}
%with $\alpha^j, L_j$ anti-Hermitian and $\tau^k, H_k$ Hermitian. 

\subsection{Block permutation examples}

Let $(\C^n)^*$ be the vector space dual to $\C^n$. For a matrix geometry with $\mathcal H_0\subset M(n,\C)\cong \C^n\otimes(\C^n)^*$ and $\mathcal A\subset M(n,\C)$, the operators $X_i\in\End(\mathcal H_0)\subset \End(\C^n)\otimes\End((\C^n)^*)$ can each be written in terms of a finite set of matrices $Y_{ij}, Z_{ij}\in M(n,\C)$, $j=1,\ldots,J$, with $[Z_{ij},a]=0$ for all $a\in\mathcal A$, by
\begin{equation}X_i\colon m\mapsto \sum_j Y_{ij}mZ_{ij},\end{equation} 
using matrix multiplication. 

This simplifies in the cases where $\mathcal H_0$ is a block permutation matrix when decomposed into irreducible bimodules (i.e., there is only one non-zero block on each row or column of the block decomposition of $\mathcal H_0$). Then the matrices $Y_{ij}$ and $Z_{ij}$ must be block diagonal to preserve the block permutation form of $\mathcal H_0$. 
Since $Z_{ij}$ commutes with the right action, it acts as a scalar in each irreducible representation of $\mathcal A$ in $(\C^n)^*$.  Therefore $mZ_{ij}=Z'_{ij}m$ for matrices $Z'_{ij}$, in which the diagonal blocks of $Z_{ij}$ are permuted. Hence $X_i$ acts as the left multiplication of the matrix $K_i=\sum_j X_{ij}Z'_{ij}\in M(n,\C)$. The matrix $K_i$ is also block diagonal.

Thus $\theta$ acts as an operator in $\mathcal H$ by 
\begin{equation}\label{eq:thetaaction}\theta (v\otimes m)=\sum_i (\omega^i v)\otimes (K_i m)\end{equation}
using matrix multiplication. Some examples of block permutation fermion spaces are examples \ref{ex:diagonal}, \ref{ex:offdiagonal} and \ref{ex:complexconj}.

%A calculation shows that $J\theta J^{-1}$ is self-adjoint, and so $D$ itself is self-adjoint. For $s$ even, \eqref{eq:evencase} implies $D\Gamma=-\Gamma D$. For $s$ odd, each $\gamma^a$ commutes with $\gamma$, so that $D\Gamma=\Gamma D$ in this case. 

%For axiom \eqref{ax:epsilonprime},
%\begin{equation} DJ=\theta J+\epsilon'J\theta=\epsilon'J\theta+J^2\theta J^{-1}=\epsilon'JD.\end{equation}
%It is worth noting that the corresponding property  $C\gamma^a=\epsilon'\gamma^a C$ is not used here, so one could in fact use a Clifford module with the opposite sign of $\epsilon'$. However this leads to the same construction for the Dirac operator, as is evident by the fact that the matrices $i\gamma^a$, with same real structure $C$, determine a Clifford module with the opposite sign for $\epsilon'$. 

\section{Fuzzy geometries}\label{sec:fuzzygeometries}

\subsection{Fuzzy spaces}\label{sec:fuzzyspaces}
The formalism of matrix geometries simplifies if the following assumptions are made. Firstly $\mathcal A$ is a simple algebra, i.e., $M(n,\R)$, $M(n,\C)$ or $M(n/2,\HH)$. Then it is assumed that $\mathcal H_0=M(n,\C)$ with $\mathcal A$ acting by matrix multiplication, expressing quaternions as $2\times2$ complex matrices so that $M(n/2,\HH)\subset M(n,\C)$. These particular cases of matrix geometries are called fuzzy spaces. The definition is the adaptation of example \ref{ex:matrixgeom} and definition \ref{def:matrixgeom} to this case and is summarised here explicitly.

 The Hilbert space $V$ is the spinor space for type $(p,q)$ gamma matrices with chirality $\gamma$ and real structure $C$. For $p+q$ even the space $V$ is irreducible, and for $p+q$ odd the chiral subspaces $V_\pm$ are irreducible. The construction uses a choice of natural number $n$.

\begin{itemize}\item $s\equiv q-p \mod 8$
\item $\mathcal A=M(n,\C), M(n,\R) \text{ or } M(n/2,\HH)$.
\item $\mathcal H=V\otimes M(n,\C)$
\item $\langle v\otimes m,v'\otimes m'\rangle=(v,v')\tr m^*m'$
\item $\rho(a)(v\otimes m)=v\otimes (am)$
\item $\Gamma(v\otimes m)=\gamma v\otimes m$
\item $J(v\otimes m)=Cv\otimes m^*$
\end{itemize}
According to the previous section, a Dirac operator is determined by products of gamma matrices $\omega^i\in\Omega$ and matrices $K_i\in M(n,\C)$ as
\begin{equation}\label{eq:diracmatrix} D(v\otimes m)=\sum_i \omega^i v\otimes (K_i m+\epsilon' m K^*_i).
\end{equation}
If $s$ is even, then the $\omega^i$ are the product of an odd number of gamma matrices. It is worth noting that this formula does not depend on which simple algebra is chosen for $\mathcal A$.

For $\epsilon'=1$ , $C\omega^i C^{-1}=\omega^i$ and so the Dirac operator can be expressed using commutators and anti-commutators, with $\alpha^j, L_j$ anti-Hermitian and $\tau^k, H_k$ Hermitian. 
\begin{equation}\label{eq:diracformula}
D(v\otimes m)
=\sum_j\alpha^jv\otimes[L_j,m]+\sum_k\tau^kv\otimes\{H_k,m\}.
\end{equation}

For $\epsilon'=-1$ , there is an extra sign in the equation. Write $\alpha^j_+, \tau^k_+$ for the even elements of $\Omega$ and $\alpha^j_-, \tau^k_-$ for the odd elements. Then
\begin{equation}\label{eq:diracformula2}\begin{split}
 D(v\otimes m)&=\sum_j\alpha_-^jv\otimes[L_j,m]+\sum_k\tau_-^kv\otimes\{H_k,m\}\\
&+\sum_l\alpha_+^lv\otimes\{L_l,m\}+\sum_r\tau_+^rv\otimes[H_r,m].\end{split}
\end{equation}

The ambiguity in $\theta$ recorded in lemma \ref{lem:ambiguity} amounts to the freedom to add terms
\begin{equation}\psi=\begin{cases}\alpha\otimes (i1) &(\epsilon'=1)\\
\alpha_-\otimes (i1)+\tau_+\otimes 1 &(\epsilon'=-1)\end{cases}
\end{equation}
to $\theta$. In each case these terms lead to a commutator in the Dirac operator, \eqref{eq:diracformula} or \eqref{eq:diracformula2}, that vanishes.

 There is certain simplification in the case of $s$ odd and irreducible $V$ due to the fact that the product of all the gamma matrices $P$ is a scalar. Multiplying by $P$ converts an odd element of $\Omega$ into an even element. Therefore one only needs to use odd elements (or only use even elements) $\omega^i$ in the definition of the Dirac operator. For $\epsilon'=-1$, $P=\pm i$, so this relates for example $\tau_+^n=P\alpha^j_-$ and $H_n=\pm iL_j$.

\begin{examples} In the following fuzzy spaces for $n\le 2$, the gamma matrices are Hermitian or anti-Hermitian according to the definitions in Examples \ref{ex:smallcliff}. 

\begin{itemize}\item Type (0,0)
\begin{equation} D=0\end{equation}
\item Type (1,0)
\begin{equation} D=\{H,\cdot\}+\gamma^1\otimes\{H_1,\cdot\} \end{equation}
\item Type (0,1)
\begin{equation} D=[H,\cdot]+\gamma^1\otimes[L_1,\cdot]\end{equation}
\item Type (2,0)
\begin{equation} D=\gamma^1\otimes\{H_1,\cdot\}+\gamma^2\otimes\{H_2,\cdot\}\end{equation}
\item Type (1,1)
\begin{equation} D=\gamma^1\otimes\{H,\cdot\}+\gamma^2\otimes[L,\cdot]\end{equation}
\item Type (0,2)
\begin{equation} \label{dirac02} D=\gamma^1\otimes[L_1,\cdot]+\gamma^2\otimes[L_2,\cdot]\end{equation}
\end{itemize}
\end{examples}

\begin{examples}
Some higher-dimensional examples are given in the Riemannian case $p=0$, so that $s=q$. All of the gamma matrices are anti-Hermitian.
\begin{itemize}\item Type (0,3)
\begin{multline}\label{dirac03} D=\{H,\cdot\}+\gamma^1\otimes[L_1,\cdot]+\gamma^2\otimes[L_2,\cdot]+\gamma^3\otimes[L_3,\cdot]\\
+\gamma^2\gamma^3\otimes[L_{23},\cdot]+\gamma^3\gamma^1\otimes[L_{31},\cdot]+\gamma^1\gamma^2\otimes[L_{12},\cdot]+\gamma\otimes\{H_{123},\cdot\}
\end{multline}
\item Type (0,4)
\begin{equation}\label{dirac04} D=\sum_i \gamma^i\otimes[L_i,\cdot]+\sum_{i<j<k}\gamma^i\gamma^j\gamma^k\otimes\{H_{ijk},\cdot\} \end{equation}
\end{itemize}
\end{examples}
The last example shows clearly the analogy with the formula for the Dirac operator on a Riemannian manifold.
% $M$ in which the spinor bundle is a trivial vector bundle, choosing an orthonormal basis of vector fields $e_i$, $ i=1,\ldots n$, the formula for the Dirac operator is determined by a set of spin connection 1-forms $\omega_{jk}$
%\begin{equation}\label{diracM} d_M=\sum_i \gamma^i\left(\partial_{e_i}-\frac12\sum_{j<k} \omega(e_i)_{jk}\gamma^j\gamma^k\right)\end{equation}
 The first term of \eqref{dirac04} is analogous to the derivative term and the second term is analogous to the spin connection term. This analogy is also apparent in the type $(0,n)$ examples for $n<4$ but new terms with products of five or more gamma matrices may appear in the fuzzy space Dirac operator for $n>4$.

\subsection{Transformations of a fuzzy space}\label{sec:fuzzytrans}
The structure of  transformations of a fuzzy space is investigated in this section. Let $U\colon\mathcal H\to\mathcal H$ be a transformation, as in \S\ref{sec:transf}. Then $U$ induces an automorphism  $\phi\colon\mathcal A\to\mathcal A$ such that
\begin{equation}U\rho(a)U^{-1}=\rho(\phi(a))\end{equation}
and, in particular, an automorphism of the center of $\mathcal A$. For $M(n,\R)$ and $M(n/2,\HH)$, the center is $\R$ and so this is trivial. However for the real algebra $M(n,\C)$ the center is $\C$ (which does have a non-trivial automorphism, complex conjugation). Since in a fuzzy space the center of $\mathcal A$ acts as a scalar, it commutes with $U$ and hence the automorphism of the center is trivial in this case also.
Then for all fuzzy spaces the Skolem-Noether theorem implies that
\begin{equation}\phi(a)=gag^{-1}
\end{equation}
for some invertible $g\in\mathcal A$. The element $g$ is not uniquely determined since one can multiply $g$ by an invertible central element of $\mathcal A$; therefore if $U\in G$, a group of transformations, then $g\in\widetilde G$, a central extension of $G$. The automorphism can now be written
\begin{equation}U\rho(a)U^{-1}=\rho(g)\rho(a)\rho(g)^{-1},\end{equation}
which implies that $U^{-1}\rho(g)$ commutes with the action of $\mathcal A$ in $\mathcal H$. From this one can see that  $U^{-1}\rho(g)J\rho(g)J^{-1}$ commutes with both the left and the right action of $\mathcal A$ and hence acts as an operator in $V$. This argument leads to
\begin{lemma}\label{lem:fuzzytrans} A transformation $U$ of a fuzzy space acts on $v\otimes m\in\mathcal H$ by
\begin{equation}\label{eq:fuzzytrans}U(v\otimes m)=hv\otimes gmg^*\end{equation}
for some unitary $g\in\mathcal A$ and unitary $h\colon V\to V$ satisfying $h\gamma=\gamma h$ and $hC=Ch$.
Conversely, any such $g$ and $h$ determine a transformation.
\end{lemma}
\begin{proof}The argument above shows that given a transformation $U$, \\$U^{-1}\rho(g)J\rho(g)J^{-1}=h^{-1}\otimes 1$, which rearranges to \eqref{eq:fuzzytrans}. The conditions for a transformation show that  $h\gamma=\gamma h$ and $hC=Ch$, and the unitarity of $U$ gives
\begin{equation}(hv,hv')\tr(gm^*g^*gm'g^*)=(v,v')\tr m^*m'
\end{equation}
for all $v$, $v'$, $m$, $m'$. This implies that $h^*h=\lambda 1$ for $\lambda>0\in\R$ and $g^*g=\mu 1$ with $\mu>0\in\R$, and $\lambda\mu^2=1$. Replacing $h$ with $\mu h$ and $g$ with $g/\sqrt\mu$ makes $h$ and $g$ unitary. The converse is straightforward.
\end{proof}
Applying a transformation to the Dirac operator, according to \eqref{eq:transD} and \eqref{eq:diracmatrix}, gives the following simple transformation of the data for the Dirac operator
\begin{equation}\omega^i\mapsto h\omega^i h^{-1},\quad K_i\mapsto g K_i g^{-1}.
\end{equation}

\subsection{Generalised fuzzy spaces}
A generalised fuzzy space is the next simplest matrix geometry after the fuzzy spaces. It is based on example \ref{ex:offdiagonal}. This matrix geometry is described explicitly as follows. The Hilbert space $V$ is the spinor space for type $(p,q)$ gamma matrices with the same conditions as in \S\ref{sec:fuzzyspaces}. The notation $M(n,m,\mathbb C)$ is used for $n\times m$ matrices with coefficients in $\mathbb C$.

\begin{itemize}\item $s\equiv q-p \mod 8$
\item $\mathcal A=\mathcal A_1\oplus\mathcal A_2$, with each $\mathcal A_i$ a simple algebra.
\item $\mathcal H=V\otimes\bigl(M(n_1,n_2,\C)\oplus M(n_2,n_1,\C)\bigr)$
\end{itemize}
These are written in matrix notation as
\begin{equation}a=\begin{pmatrix}a_1& .\\ .&a_2\end{pmatrix}\in\mathcal A,\quad \xi=\begin{pmatrix}.&\xi_1\\\xi_2& .\end{pmatrix}\in\mathcal H.\end{equation}
\begin{itemize}
\item $\left\langle\begin{pmatrix}.&v\otimes m_1\\v\otimes m_2& .\end{pmatrix},
\begin{pmatrix}.&v'\otimes m'_1\\v'\otimes m'_2& .\end{pmatrix}\right\rangle=
(v,v')\tr (m_1^*m_1'+m_2^*m_2')$
%\item $\langle v\otimes (m_1,m_2),v'\otimes (m_1',m_2')\rangle=(v,v')\tr (m_1^*m_1'+m_2^*m_2')$
\item  $\rho\begin{pmatrix}a_1& .\\ .&a_2\end{pmatrix}
\left\langle\begin{pmatrix}.&v\otimes m_1\\v\otimes m_2& .\end{pmatrix}\right\rangle=
\left\langle\begin{pmatrix}.&v\otimes a_1 m_1\\v\otimes a_2m_2& .\end{pmatrix}\right\rangle$
%\item $\rho(a_1,a_2)\bigl(v\otimes (m_1,m_2)\bigr)=v\otimes (a_1m_1,a_2m_2)$
\item $\Gamma\begin{pmatrix}.&v\otimes m_1\\v\otimes m_2& .\end{pmatrix}=
\begin{pmatrix}.&\gamma v\otimes m_1\\\gamma v\otimes m_2& .\end{pmatrix}$
%\item $\Gamma\bigl(v\otimes(m_1,m_2) \bigr)=\gamma v\otimes (m_1,m_2)$
\item $J\begin{pmatrix}.&v\otimes m_1\\v\otimes m_2& .\end{pmatrix}=
\begin{pmatrix}.&Cv\otimes m_2^*\\Cv\otimes m_1^*& .\end{pmatrix}$
%\item $J\bigl(v\otimes (m_1,m_2)\bigr)=Cv\otimes (m_2^*,m_1^*)$
\end{itemize}
The Dirac operator is determined by the same formula \eqref{eq:diracmatrix} as for the fuzzy spaces, with the matrices $K_i$ block diagonal. Thus
\begin{equation}\label{eq:generalisedD} D=\begin{pmatrix}D_1&.\\.& D_2\end{pmatrix}
\end{equation}
and so $D_2=\epsilon'J^{-1}D_1J$, with both $D_1$ and $D_2$ self-adjoint. Thus if $\lambda$ is an eigenvalue of $D_1$, then $\epsilon'\lambda$ is an eigenvalue of $D_2$. Splitting the operators $K_i$ into anti-hermitian and hermitian generators, these are written
\begin{equation}L_j=\begin{pmatrix}L_{(1)j}&.\\.& L_{(2)j}\end{pmatrix},
\quad H_k=\begin{pmatrix}H_{(1)k}&.\\.& H_{(2)k}\end{pmatrix}
\end{equation}
with the parentheses used to avoid the notational clash of different types of indices. The Dirac operator is again given by \eqref{eq:diracformula} and \eqref{eq:diracformula2}, giving for $\epsilon'=1$
\begin{equation*}%\label{eq:dirac1}
D_1(v\otimes m_1)
=\sum_j\alpha^jv\otimes(L_{(1)j}m_1-m_1L_{(2)j})+\sum_k\tau^kv\otimes( H_{(1)k}m_1+m_1H_{(2)k})
\end{equation*}
\begin{equation}\label{eq:dirac1}
D_2(v\otimes m_2)
=\sum_j\alpha^jv\otimes(L_{(2)j}m_2-m_2L_{(1)j})+\sum_k\tau^kv\otimes( H_{(2)k}m_2+m_2H_{(1)k})
\end{equation}
which is essentially the same equation but taking into account that the `representation' of $L_j$ and $H_k$ on the left and right differ. A similar modification of the formula holds for $\epsilon'=-1$.

A similar argument to lemma \ref{lem:fuzzytrans} show that for $n_1\ne n_2$, a transformation of a generalised fuzzy space has the form \eqref{eq:fuzzytrans}, with
\begin{equation}g=\begin{pmatrix}g_{1}&.\\.& g_{2}\end{pmatrix}.
\end{equation} 
The fermions transform by
\begin{equation}U\begin{pmatrix}.&v\otimes m_1\\v\otimes m_2& .\end{pmatrix}=
\begin{pmatrix}.& hv\otimes g_1m_1g_2^*\\ hv\otimes g_2m_2g_1^*& .\end{pmatrix}.\end{equation}
For $n_1=n_2$, there are also transformations that interchange the two blocks.

\section{Spherical fuzzy spaces}\label{sec:spheres}

This section is devoted to explaining how the fuzzy sphere, and various generalisations of it, can be presented as a real spectral triple. These spaces have the symmetry group $\SU(2)$, which is the same as the symmetry group of the spectral triple for the commutative 2-sphere $S^2$.

The Grosse-Presnajder operator \eqref{GP} has this symmetry but
 is not an example of a type $(0,2)$ fuzzy space \eqref{dirac02}. However, it is an example of a type $(0,3)$ fuzzy space  \eqref{dirac03}, which does not require a chirality operator, since $s=3$. This is a somewhat unsatisfactory situation since it is desirable for a fuzzy analogue of the commutative 2-sphere to have $s=2$. 

The solution to this problem that is proposed here is to augment the type $(0,3)$ Clifford module with a Hermitian generator to give a type $(1,3)$ Clifford module, so that the KO-dimension $s=3-1=2 \mod 8$, as expected for a $2$-sphere.
 It is a convenient notation to call the sole Hermitian gamma matrix $\gamma^0$ so that the gamma matrices are $\gamma^1, \gamma^2, \gamma^3,\gamma^0$.
These are required to form an irreducible Clifford module and the space the gamma matrices act in is $V\cong\C^4$. 

 To construct a fuzzy space of type $(1,3)$ having a  Dirac operator with spherical symmetry, consider an  $n$-dimensional representation of $\Spin(3)=\SU(2)$, the double cover of  $\SO(3)$. Let $L_{jk}$, $j<k=1,2,3$ be standard generators of the Lie algebra $\so(3)$ in this representation. This representation may be reducible, though in the simplest case, in \S\ref{sec:fuzzysphere}, it will be irreducible.  These matrices satisfy
\begin{equation}\label{eq:standardgen}[L_{jk},L_{lm}]=\delta_{kl}L_{jm}-\delta_{km}L_{jl}-\delta_{jl}L_{km}+\delta_{jm}L_{kl}
\end{equation}
using the relations $L_{jk}=-L_{kj}$. The Hilbert space of the spectral triple is $V\otimes M(n,\C)$ and the Dirac operator is defined to be
\begin{equation}\label{eq:fuzzydirac} 
D=\gamma^0+\sum_{j<k=1}^3\;
\gamma^0\gamma^j\gamma^k\otimes [L_{jk},\cdot\,]
\end{equation}
The product $\gamma^0\gamma^j\gamma^k$ is odd and anti-Hermitian, so this formula is an example of a Dirac operator according to \eqref{eq:diracformula}.

Using the product representation from example \ref{ex:product10}, the gamma matrices can be written
\begin{equation}\label{eq:gamma13}\gamma^0=\begin{pmatrix}0&1\\1&0\end{pmatrix},\quad \gamma^a=\begin{pmatrix}0&i\sigma^a\\-i\sigma^a&0\end{pmatrix}, \quad a=1,2,3
\end{equation}
with $\sigma^a$ gamma matrices for an irreducible type $(0,3)$ Clifford module acting in vector space $V'=\C^2$, with real structure $C'$.
Then the Dirac operator becomes the block matrix
\begin{equation}\label{eq:diraceven} D=\begin{pmatrix}0&\widetilde d\,\\ \widetilde d&0\end{pmatrix}
\end{equation}
with off-diagonal blocks the Grosse-Presnajder operator, repeated here,
\begin{equation}\label{eq:nc-mod-dirac}
\widetilde d=1+\sum_{j<k}\sigma^j\sigma^k\otimes [L_{jk},\cdot\,]
\end{equation} 
The chirality operator is
\begin{equation}\label{eq:fuzzychirality}\Gamma=\gamma\otimes 1=\sigma\begin{pmatrix}1&0\\0&-1\end{pmatrix},\end{equation}
using the scalar $\sigma=\pm1$, which is the chirality for the $\sigma^a$.
Note that $\widetilde d$ commutes with the real structure $J'=C'\otimes *$ constructed using the real structure $C'$ for the $\sigma^a$. 

By a change of basis using the matrix
\begin{equation}W=\frac1{\sqrt2}\begin{pmatrix}1 &1\\-1&1\end{pmatrix},\end{equation}
in which each $1$ is a $2n\times 2n$ identity matrix, the Dirac operator is block-diagonalised as
\begin{equation}WDW^{-1}=\begin{pmatrix}\widetilde d&0\\ 0&-\widetilde d\,\end{pmatrix}.\end{equation}
In this basis $\gamma^0$ becomes
\begin{equation}W\gamma^0W^{-1}=\begin{pmatrix}1&0\\ 0&-1\end{pmatrix}.\end{equation}
but the chirality operator is off-diagonal
\begin{equation}W\Gamma W^{-1}=-\sigma\begin{pmatrix}0&1\\1&0\end{pmatrix}.\end{equation}
So although the Dirac operator is well-defined on the $\gamma^0=1$ eigenspace as the top-left block, the chirality operator projected onto this subspace vanishes. This explains the relation of the Grosse-Presnajder operator to this spectral triple.

The group $\SU(2)$ acts as a group of transformations. An element 
$h\in\SU(2)$ acts in $V$ and determines a transformation
$U$ according to formula \eqref{eq:fuzzytrans} of lemma \ref{lem:fuzzytrans}, with $g(h)$ the matrix representation of $\SU(2)$ in $\C^n$ that has the Lie algebra action used above. Under this action, the Dirac operator \eqref{eq:fuzzydirac} is mapped to $UDU^{-1}$. This determines the adjoint action on both $L_{jk}$ and the matrices
\begin{equation}\label{eq:SpinOperatordef}\SpinOperator_{jk}=-\frac14(\sigma^j\sigma^k-\sigma^k\sigma^j),
\end{equation}
which are also standard generators for $\so(3)$, obeying similar relations to \eqref{eq:standardgen}, replacing $L_{jk}$ with $\SpinOperator_{jk}$. (Note that the three-dimensional indices here can be raised and lowered with the identity matrix $\delta_{jk}$ and so there is no particular significance in whether subscripts or superscripts are used.) Thus $D$ is invariant and $\SU(2)$ is a group of symmetries. It is worth noting that the non-trivial symmetries are not inner automorphisms of the spectral triple.

The operator $\gamma^0$ commutes with the Dirac operator and the action of $\SU(2)$ and so is useful is classifying the eigenvectors of the Dirac operator. 

The spectrum of $\widetilde d$ can be analysed in terms of Casimirs for the Lie algebra $\so(3)$.
Define the dot product of Lie algebra generators as, for example, 
\begin{equation}\SpinOperator\centerdot  L=\sum_{j<k}\SpinOperator_{jk}L_{jk}.\end{equation}
 Putting $c_1=-\SpinOperator\centerdot\SpinOperator$, $c_2=-\Lambda\centerdot \Lambda$, with $\Lambda_{jk}$ denoting the adjoint action of $L_{jk}$ on matrices, $m\mapsto [L_{jk},m]$, and $c=-(\SpinOperator+\Lambda)\centerdot(\SpinOperator+\Lambda)$, this gives
\begin{equation}\label{dfromcasimirs}\widetilde d=-2\SpinOperator\centerdot\Lambda+1=c-c_1-c_2+1.\end{equation}

\subsection{The fuzzy sphere}\label{sec:fuzzysphere}
 The simplest case is where the action $g(h)$ of $\SU(2)$ is irreducible and is called the fuzzy sphere.
An irreducible representation of $\Spin(3)\cong\SU(2)$  is determined by a half-integer $l\in\frac12\Z$, $l\ge0$, called the spin of the representation. The representation has dimension $2l+1$ and Casimir $l(l+1)$. Every representation is isomorphic to its complex conjugate. The spinor space for the modified Dirac operator $V'\cong\C^2$ carries the irreducible representation $l=\frac12$ and so $c_1=\frac34$. 

For the fuzzy sphere determined by the irreducible representation $l=\frac12(n-1)$ on $\C^n$, the adjoint representation on $M(n,\C)$ is $l\otimes l\cong \oplus k$, with $k\in \{0,1,\ldots 2l\}$. Hence the representation on $\C^2\otimes M(n,\C)$ is $\frac12\otimes\bigl(\oplus k\bigr)=\oplus j$, with
$j=k\pm\frac12$ for $k\ne0$ and $j=\frac12$ for $k=0$. Thus $j\in\{\frac12,\frac32,\ldots,n-\frac12\}$ each with multiplicity two (corresponding to $\pm$) except for the top spin $ n-\frac12$, which has multiplicity one (corresponding to the sign $+$). Then $c_2=k(k+1)$ and $c=j(j+1)$, leading to $\widetilde d=\pm(j+\frac12)$, again with the same sign $\pm$. Thus the spectrum of $\widetilde d$ is
\begin{equation}\label{eq:specfuzzy2}\pm1,\;\pm2,\ldots,\;\pm(n-1),\;+n\end{equation}
each eigenvalue corresponding to a single irreducible of $\Spin(3)$. This spectrum is not symmetrical about zero. 

The Dirac operator $D$ acts on $\mathcal H\cong \C^4\otimes M(n,\C)$ and has the spectrum of $\widetilde d\oplus-\widetilde d$, as shown in the following table.
$$
\begin{tabular}{|l|cccccccccc|}
  \hline
  $\gamma^0$ \rule{0pt}{15pt}&$-1$&$-1$&\dots&$-1$&$-1$&1&1&\dots&1&1\\[5pt] %space below text
%\hline
$j$\rule{0pt}{15pt}
&$n-\frac12$&$n-\frac32$&\dots&$\frac32$&$\frac12$
&$\frac12$&$\frac32$&\dots&$n-\frac32$&$n-\frac12$\\[5pt]
\hline
$D$\rule{0pt}{15pt}%makes line 15pt high
&$-n$&$\mp(n-1)$&\dots&$\mp2$&$\mp1$
&$\pm 1$&$\pm2$&\dots&$\pm(n-1)$&$n$\\[5pt]
 \hline
\end{tabular}
$$
 The multiplicity of each simultaneous eigenvalue of $D$ and $\gamma^0$ is the dimension of the irreducible representation of $\Spin(3)$ that is shown in the second row.  Hence each irreducible of $\Spin(3)$ that appears in the table occurs four times in total, except for  $n-\frac12$, which occurs twice.

\subsection{The generalised fuzzy sphere}\label{sec:gensphere}

Generalised fuzzy spaces can be defined using the same formula  \eqref{eq:fuzzydirac} for the Dirac operator. They can be written in blocks in spinor space in the same way as for a fuzzy case, \eqref{eq:diraceven}. Decomposing into the two blocks in the space of matrices (as in \eqref{eq:generalisedD}) gives
\begin{equation}\widetilde d=\begin{pmatrix}\widetilde d_1&.\\.&\widetilde d_2\end{pmatrix},
\end{equation}
with each $\widetilde d_i$ again expressed in terms of Casimir operators, according to \eqref{dfromcasimirs}, the difference being that $\Lambda$ is no longer the adjoint action of the Lie algebra on the matrices. Written as a left action, it acts on $m_1$ as the operators of the Lie algebra associated to the representation $g_1\otimes \overline g_2$ of $\SU(2)$, $\overline g_2$ denoting the complex conjugate matrix. Similarly, it acts on $m_2$ as the Lie algebra for the representation $g_2\otimes \overline g_1$. The two operators $\widetilde d_1$, $\widetilde d_2$ are related by $J'\widetilde d_1=\widetilde d_2 J'$, and so have the same spectrum. This can also be seen from the fact that $J'$ intertwines the Casimir operators. The two matrix blocks of the Dirac operator are determined by the spinor block decomposition
\begin{equation}\label{eq:diracgensphere} D_1=\begin{pmatrix}0&\widetilde d_1\,\\ \widetilde d_1&0\end{pmatrix},\quad\quad  D_2=\begin{pmatrix}0&\widetilde d_2\,\\ \widetilde d_2&0\end{pmatrix}.
\end{equation}

 The simplest case is where the actions of $g_{1}(h)$ and $g_{2}(h)$ are irreducible. This is called the generalised fuzzy sphere. Suppose the irreducible representations have spin $l_1=\frac12(n_1-1)$ and $l_2=\frac12(n_2-1)$. The decomposition of $M(n_1,n_2,\C)$ is $l_1\otimes l_2\cong \oplus k$ with $k\in\{|l_1-l_2|,|l_1-l_2|+1,\ldots,l_1+l_2\}$. Again, $j=k\pm\frac12$, $j\ge0$, and
the spectrum of $\widetilde d_1=\pm(j+\frac12)$ for $n_1\ne n_2$ is therefore
\begin{equation}-\frac12(|n_1-n_2|),\;\pm(\frac12|n_1-n_2|+1),\ldots,\;\pm(\frac12(n_1+n_2)-1),\;+\frac12(n_1+n_2),\end{equation}
each eigenspace being the irreducible representation of $\SU(2)$ of spin $j$.
For $n_1= n_2$ the first term is absent, so the spectrum is the same as \eqref{eq:specfuzzy2}. This can also be understood as the first term belonging to the zero-dimensional representation in this case. The spectrum for $\widetilde d_2$ is the same as $\widetilde d_1$. 

Since $D_1$ is isomorphic to $\widetilde d_1\oplus-\widetilde d_1$ by a change of basis, the spectrum of $D_1$ is given in the following table, in which $m=|n_1-n_2|/2$ and $n=(n_1+n_2)/2$.
$$
\begin{tabular}{|l|cccccccc|}
  \hline
  $\gamma^0$ \rule{0pt}{15pt}&$-1$&$-1$&\dots&$-1$&1&\dots&1&1\\[5pt] %space below text
$j$\rule{0pt}{15pt}
&$n-\frac12$&$n-\frac32$&\dots&$m-\frac12$
&$m-\frac12$&\dots&$n-\frac32$&$n-\frac12$\\[5pt]
\hline
$D_1$\rule{0pt}{15pt}%makes line 15pt high
&$ -n$&$\mp(n-1)$&\dots&$m$
&$-m$&\dots&$\pm(n-1)$&$n$\\[5pt]
  \hline
\end{tabular}
$$
Using $D=D_1\oplus D_2$, the spectrum of $D$ is just twice the spectrum for $D_1$.

\paragraph{Reducible fuzzy cases.}
If the action $g(h)$ for a fuzzy space is reducible, the Dirac operator \eqref{eq:fuzzydirac} can be decomposed into a sum of these irreducible cases, which are the fuzzy sphere and the generalised fuzzy sphere. For a given Dirac operator $D$, the algebra $\mathcal A$ can always be replaced by a subalgebra and $D$ remains a Dirac operator. The notion of a symmetry group does however depend on the algebra, since it acts as automorphisms of the algebra. Thus for the fuzzy sphere with $\mathcal A=M(n,\C)$, if the representation of $\SU(2)$ on $\C^n$ splits into irreducibles of dimension $n_i$, $\sum_i n_i=n$, the algebra can be replaced with the subalgebra $\oplus_i\mathcal A_i\subset M(n,\C)$, with $\mathcal A_i=M(n_i,\C)$.
% the case of a complex representation, $M(n_i,\R)$ in the case of a real representation, and $M(n_i/2,\HH)$ in the case of a quaternionic representation.
 The Hilbert space of the fuzzy sphere also decomposes into blocks according to the decomposition  $\C^n=\C^{n_1}\oplus\C^{n_2}\oplus\ldots$. The diagonal blocks in the matrices of $\mathcal H_0$ are again fuzzy spheres, and a pair of off-diagonal blocks related by $J$ is a generalised fuzzy sphere. Thus the analysis of the Dirac operator of the fuzzy sphere for a general representation reduces to the analysis of the irreducible cases. Note that if $\mathcal A$ is one of the other two simple algebras, there are some restrictions on the possible decompositions.

\subsection{Commutative analogues}\label{sec:commgeom}

The purpose of this section is to exhibit commutative analogues of the fuzzy sphere and generalised fuzzy sphere. These analogues are constructed so that, in a suitable basis, the Dirac operator of the commutative version is simply obtained from the fuzzy version by replacing the matrix commutators with an appropriate Lie derivative on the sphere.  

\paragraph{Dirac operator on the sphere.} The $2$-sphere $S^2=\{x=(x_1,x_2,x_3)\in\R^{3}\mid x^2=1\}$, with a chosen orientation,  has a spin structure determined by the embedding in $\R^{3}$ and thus a spinor bundle $E\to S^2$ determined by an irreducible Clifford module of type $(0,2)$. 

The tangent bundle of a sphere can be trivialised by adding the normal bundle, in other words, enlarging the structure group to $\SO(3)$. Accordingly, the spinor bundle can be trivialised  (as a vector bundle) by extending the Clifford module to one of type $(0,3)$, which admits $\SU(2)$ transformations. This can be thought of as the pull-back to $S^2$ of the trivial bundle of spinors in $\R^3$. 

This idea is taken as the definition here: the spinor bundle $E$ is defined as the trivial bundle  $\C^2\times S^2\to S^2$, with $\C^2$ a fixed type $(0,3)$ irreducible Clifford module having gamma matrices $\sigma^a$, $a=1,2,3$. (The letter $\sigma$ is used to emphasise that these are gamma matrices in one more dimension.) A useful notation is that a vector $w=(w_1,w_2,w_3)\in\R^{3}$ determines a matrix $w\cdot \sigma=\sum_{a=1}^{3} w_a\sigma^a$ acting in $\C$. At a point $x\in S^2$, each tangent vector $v$  determines the matrix $v\cdot\sigma$, so an oriented orthonormal basis of the tangent space  generates a Clifford module of type $(0,2)$. The Clifford structure varies with $x$ in a non-trivial way. 
The normal vector $x$ determines the chirality operator, 
\begin{equation}\gamma_S=i\, x\cdot\sigma.\end{equation}

The sphere has vector fields
\begin{equation} X_{jk}=x_j\frac{\partial}{\partial x_k}-x_k\frac{\partial}{\partial x_j}.
\end{equation}
that, for $j<k$, form a standard basis of the Lie algebra of $\SU(2)$, satisfying the same relations as the $L_{jk}$ in \eqref{eq:standardgen}. If $F$ is another complex vector bundle over $S^2$ with fibre $F_x$ at $x$, then write $E\otimes F$ for the bundle with fibre $\C^2\otimes F_x$. Given a covariant derivative $\nabla$ on $E\otimes F$, the corresponding Dirac operator is defined by the formula
\begin{equation}d_S=- (x\cdot\sigma)\sum_{j<k}\sigma^j\sigma^k\,\nabla_{X_{jk}}.\end{equation}
That this is the correct formula can be seen by looking at the point $x=(0,0,1)$. There, $X_{12}=0, X_{13}=-\partial/\partial x_1, X_{23}=-\partial/\partial x_2$, so 
\begin{equation}d_S=
-\sigma^3\bigl(\sigma^1\sigma^3\,\nabla_{-\partial/\partial x_1}+\sigma^2\sigma^3\,\nabla_{-\partial/\partial x_2}\bigr)=
\sigma^1\,\nabla_{\partial/\partial x_1}+\sigma^2\,\nabla_{\partial/\partial x_2}.\end{equation}
A similar formula holds at other points by taking suitable bases. The Dirac operator $d_S$ is a self-adjoint operator \cite{WolfEssential}.

\paragraph{Analogue of the fuzzy sphere.} The geometric model for the analogue of the fuzzy sphere is based on the Dirac operator on $E$ (i.e., $F_x=\C$ in the above discussion).

An element $g\in\SU(2)$ acts on $E$ by the simultaneous action of the matrix $g$ on $\C^2$ and the rotation of $S^2$ determined by $g$. This action is compatible with the Clifford structure. The corresponding Lie derivative on sections of $E$, written as $\C^2\otimes \C(S^2)$,  in  the direction $X_{jk}$ is
\begin{equation} \Lie(X_{jk})=\SpinOperator_{jk}+X_{jk}.
\end{equation}

The canonical choice of covariant derivative $\nabla$ on $E$ is the Levi-Civita connection, which is $\SU(2)$-invariant. This is obtained by projecting the covariant derivative from $\R^3$ into the surface, and gives
 \begin{equation} \nabla_{X_{jk}}=X_{jk}+\sum_l -x_j\SpinOperator_{kl}x_l+x_k\SpinOperator_{jl}x_l.\end{equation}

The Dirac operator defined by $\nabla$ forms a real spectral triple. It has eigenvalues \cite{CH,Trautman:1995fr}
\begin{equation}\lambda=\pm\left(r+1\right),\quad r=0,1,2,\ldots\end{equation}
and multiplicities $\mu=2(r+1)$. The eigenspaces form an irreducible representation of $\SU(2)$ of spin $j=r+1/2$.

The Dirac operator can be written in terms of the Lie derivative by computing the difference between the two. The  operator
\begin{equation}H_{jk}= \Lie(X_{jk})-\nabla_{X_{jk}}\end{equation}
is a linear map on each fibre $\C^2$, and is given explicitly by
\begin{equation}\label{eq:Hformula} H_{jk}=\SpinOperator_{jk}+\sum_l x_j\SpinOperator_{kl}x_l-x_k\SpinOperator_{jl}x_l.\end{equation}
Then a calculation gives
\begin{equation}\sum_{j<k}\sigma^j\sigma^kH_{jk}=\frac12\,,\quad\sum_{j<k}\sigma^j\sigma^k\SpinOperator_{jk}=\frac32,\end{equation}
which leads to the following formula \cite{trautman2008} for the Dirac operator,
\begin{equation}%\label{eq:modified-dirac}
 d_S=-\bigl(x\cdot\sigma\bigr)\left(1+\sum_{j<k}\sigma^j\sigma^k X_{jk}\right).\end{equation}
The relation between the Dirac operator and the chirality operator 
\begin{equation}\label{eq:sphereDgamma}\gamma_Sd_S=-d_S\gamma_S,\end{equation}
can be checked from this formula by direct calculation.

The modified Dirac operator is defined by
\begin{equation}\label{eq:modified-dirac}
\widetilde d_S=\bigl(x\cdot\sigma\bigr)d_S=1+\sum_{j<k}\sigma^j\sigma^k X_{jk},\end{equation}
which is the analogue of the Grosse-Presnajder operator with $[L_{jk},\cdot\,]$ replaced with its Lie-algebraic analogue $X_{jk}$.
This equation can be solved to give the Dirac operator
\begin{equation}d_S=-\bigl(x\cdot\sigma\bigr)\widetilde d_S=i\gamma_S\widetilde d_S. \end{equation}
Since $s=2$, the operators $d_S$ and $\widetilde d_S$ are related by chiral rotation, as in \eqref{eq:modified}. This is $\widetilde d_S=R d_S R^{-1}$ with $R=e^{\pi (x\cdot\sigma)/4}$. Hence $d_S$ and $\widetilde d_S$ have the same spectrum.

Using the Casimir operators $c_1=-\SpinOperator\centerdot\SpinOperator$, $c_2=-X\centerdot X$, $c=-(\SpinOperator+X)\centerdot(\SpinOperator+X)$, the modified Dirac operator can be written
\begin{equation}\label{eq:dirac-casimir}\widetilde d_S=-2\SpinOperator\centerdot X+1=c-c_1-c_2+1.\end{equation}
This formula is the analogue of  \eqref{dfromcasimirs} for the Grosse-Presnajder operator in the fuzzy case, the only difference being that a different list of irreducible representations occurs. The correspondence guarantees that when the same representations occur, then the eigenvalues of the modified Dirac operators in the two cases coincide.

Comparing the spectrum of $d_S$ with the spectrum of the fuzzy sphere $D$, one can see that the fuzzy sphere spectrum corresponds to two Dirac fields on $S^2$, one cutoff after eigenvalues $\pm n$, the other cutoff after eigenvalues $\pm(n-1)$.

Therefore to get an analogue of the Dirac operator it is necessary to double the fermions in the commutative case. Using the opposite sign for the Dirac and chirality operators, one gets the following block matrix operators for two fermion fields on the sphere
\begin{equation}D_S=\begin{pmatrix}d_S&0\\0&-d_S\end{pmatrix},\quad \Gamma_S=\sigma\begin{pmatrix}\gamma_S&0\\0&-\gamma_S\end{pmatrix}.
\end{equation}
with the constant $\sigma$ being either $1$ or $-1$.
The point of this construction is that the $+1$ eigenspace of $\Gamma_S$ consists of one copy of the $+1$ eigenspace of $\gamma_S$ and one copy of the $-1$ eigenspace of $\gamma_S$,  fitting together to form a trivial bundle. The same is true of the $-1$ eigenspace of $\Gamma_S$. 

This remark can be formalised in the following way. Let $\pi_\pm=\frac12(1\pm\gamma_S)$ be the projectors onto the $\pm1$ eigenspaces of $\gamma_S$. Define the unitary operator
\begin{equation}T=\begin{pmatrix}\pi_+&-i\pi_-\\i\pi_-&\pi_+\end{pmatrix}
\end{equation}
to change basis in the Hilbert space.
Then a calculation using \eqref{eq:sphereDgamma} shows that
\begin{equation}TD_ST^{-1}=\begin{pmatrix}0&\widetilde d_S\\\widetilde d_S&0\end{pmatrix},\quad
T\Gamma_S T^{-1}=\sigma\begin{pmatrix}1&0\\0&-1\end{pmatrix},
\end{equation}
which is exactly the commutative analogue of the fuzzy sphere Dirac operator \eqref{eq:diraceven} and its chirality operator \eqref{eq:fuzzychirality}.
This means that one can construct a bundle over $S^2$ with two Dirac fermion fields by simply replacing $[L_{jk},\cdot\,]$ with $X_{jk}$ in \eqref{eq:fuzzydirac}.

\paragraph{Analogue of generalised fuzzy sphere.}
The commutative analogue of the generalised fuzzy sphere is a construction involving monopoles. The vector bundle $E$ decomposes into two complex line bundles $L^+$ and $L^-$ given, at each point, by the $\pm1$ eigenspaces of the chirality $\gamma_S$. The monopole with charge $\kappa$ is the line bundle $L^\kappa$ over $S^2$ with fibre $(L^+)^\kappa=(L^-)^{-\kappa}$, and is toplogically non-trivial if $\kappa\ne0$. 
Since $L^\pm$ are preserved by the Lie derivative and the covariant derivative, they have both a Levi-Civita connection and an action of $\SU(2)$, and so therefore does the general monopole $L^\kappa$.

The spectrum of the Dirac operator $d_S$ on $E\otimes L^\kappa $ is computed in \cite{Jayewardena:1988td}. The eigenvalues are
\begin{equation}\label{eq:monpoleeigenvalues}0,\,\pm\sqrt{r(r+|\kappa|)}\quad\quad r=1,2,\ldots
\end{equation}
The two eigenspaces for the integer $r\ge1$ are each an irreducible representation of spin $j=|\kappa|/2-1/2+r$. The eigenspace of eigenvalue $0$ forms a single irreducible of spin  $j=|\kappa|/2-1/2$ and the chirality operator $\gamma_S$ takes the sign $-\kappa/|\kappa|$ on this subspace. Thus the index of $d_S$ is $-\kappa$. 

The other main feature of  $d_S$ is that for $\kappa\ne0$, there is no real structure. This is because the complex conjugate of the bundle $L^\kappa$ is $L^{-\kappa}$, so to get a real spectral triple one must take the direct sum of the two bundles, $E\otimes( L^{\kappa}\oplus L^{-\kappa})$, the operator $J_S$ mapping from one bundle to the other. This is the counterpart of the two blocks in the matrix case.

Comparing the spectrum of $d_S$ to the spectrum of $D_1$ in \S\ref{sec:gensphere}, the set of irreducible representations in $D_1$ is the same as two copies of $d_S$ with the same value of $\kappa$, providing that $m=|\kappa|/2$ and the spins are cut off at $j=n-\frac12$ for one copy and $j=n-\frac32$ for the other. However, the eigenvalues do not match.

This defect is remedied by adding a mass matrix to the Dirac operator after doubling the fermions. The definition of the commutative Dirac and chirality operators that are the analogues for the generalised fuzzy case is 
\begin{equation}D_S=\begin{pmatrix}d_S&\frac12\kappa\\\frac12\kappa&-d_S\end{pmatrix},\quad \Gamma_S=\sigma\begin{pmatrix}\gamma_S&0\\0&-\gamma_S\end{pmatrix}
\end{equation}
with $\sigma=\pm1$. These obey
\begin{equation}D_S\Gamma_S=-\Gamma_SD_S\end{equation}
as a consequence of \eqref{eq:sphereDgamma}, which still holds in the monopole case. Using the same basis change $T$ as for the case of the fuzzy sphere results in
\begin{equation}\label{eq:commgen}TD_ST^{-1}=\begin{pmatrix}0&\widetilde d_S+\frac12\kappa\gamma_S\\\widetilde d_S+\frac12\kappa\gamma_S&0\end{pmatrix},\quad
T\Gamma_S T^{-1}=\sigma\begin{pmatrix}1&0\\0&-1\end{pmatrix},
\end{equation}
with again $\widetilde d_S=(x\cdot\sigma)d_S$.

The rest of this section is devoted to showing that these operators are the direct analogues of $D_1$ and $\Gamma_1$ of  \S\ref{sec:gensphere}. In general, the Lie derivative on $E\otimes F$ can be written as
\begin{equation}\Lie^{E\otimes F}(X_{jk})=\SpinOperator_{jk}\otimes 1+1\otimes \Lie^F(X_{jk})\end{equation}
and
\begin{equation}H_{jk}^{E\otimes F}=H_{jk}^E\otimes 1+1\otimes H_{jk}^F.
\end{equation}
Here, the bundle $F=L^\kappa$, a line bundle, and so $H_{jk}^F$ is a function on $S^2$. The modified Dirac operator on $E\otimes F$ can be written
\begin{equation}\widetilde d_S=1+\sum_{j<k}\sigma^j\sigma^k\left(\Lie^F(X_{jk})-H^F_{jk}\right).\end{equation}
To get a formula for $H^F_{jk}$, rewrite $H^E_{jk}=H_{jk}$ from \eqref{eq:Hformula} using the 3-dimensional permutation symbol $\epsilon_{jkl}$ and the identity $\sum_j 2\sigma\sigma^j\epsilon_{jkl}=\sigma^k\sigma^l-\sigma^l\sigma^k$, giving
\begin{equation}H^E_{jk}=-\frac12\sigma(x.\sigma)\sum_l x_l\epsilon_{jkl}.\end{equation}
Then $H^F$ is $\kappa$ times the eigenvalue of this matrix on the subspace $\gamma_S=1$, i.e.,
\begin{equation}H^F_{jk}=\frac i2\kappa\sigma\sum_l x_l\epsilon_{jkl}.\end{equation}
Hence
\begin{equation}\sum_{j<k}\sigma^j\sigma^kH^F_{jk}=\frac12\kappa\gamma_S,\end{equation}
and the combination that appears in \eqref{eq:commgen} is
\begin{equation}\label{eq:monopoleoperator}\widetilde d_S+\frac12\kappa\gamma_S=1+\sum_{j<k}\sigma^j\sigma^k\Lie^F(X_{jk})=c-c_1-c_2+1,\end{equation}
which is exactly analogous to the fuzzy operator $\widetilde d_1$ with $m=|\kappa|$. This has the eigenvalues $\pm(j+\frac12)$ by the same argument as the fuzzy case. As a check, the eigenvalues $\lambda$ of $\widetilde d_S$, and hence $d_S$, can be verified. Squaring both sides of \eqref{eq:monopoleoperator} gives
\begin{equation} \lambda^2+\frac14\kappa^2=(j+\frac12)^2\end{equation}
in agreement with \eqref{eq:monpoleeigenvalues}.

In summary, by a suitable change of basis one sees that the operator $D_S$ on the monopole bundle is obtained from the generalised fuzzy sphere Dirac operator $D_1$ by replacing the commutators with Lie derivatives. Hence the eigenvalues are exactly the same, except that there is no maximum cutoff on the spin of the representations in the commutative case. Both operators have index equal to $0$. Similarly, the commutative analogue of the operator $D_2$ is the operator $D_S$ constructed from bundle $F=L^{-\kappa}$.

\section{Acknowledgement}
Support from the Particle Physics Theory Consolidated Grant ST/L000393/1 awarded by the Science and Technology Facilities Council is acknowledged.

\end{document}